\documentclass[12pt]{article}
\usepackage{amsmath}
\usepackage{amssymb, amscd, amsthm,amsfonts}
\usepackage[all]{xy}
\usepackage[dvips]{graphicx}
\usepackage{verbatim}
\newtheorem{theo}[]{{\emph{Theorem}}}

\newtheorem{lemma}[]{{\emph{Lemma}}}

\theoremstyle{remark}
\newtheorem*{remark}{\textbf{Remark}}

\theoremstyle{definition}

\newcommand{\ra}{\rightarrow}

\newcommand{\calF}{\mathcal{F}} 

\newcommand{\bF}{\mathbb{F}}

\newcommand{\cC}{\mathcal{C}}

\newcommand{\al}{\alpha}
\newcommand{\be}{\beta}
\newcommand{\ga}{\gamma}
\newcommand{\veps}{\varepsilon}

\newcommand{\om}{\omega}

\DeclareMathOperator{\Tra}{\mathrm Tr}

\begin{document}

\title{Exponential Sums, Cyclic Codes and Sequences: the Odd Characteristic Kasami Case} \maketitle

\begin{center}
$\mathrm{Jinquan\;\;Luo\qquad\quad Yuansheng \;\;Tang \quad\qquad
and\qquad Hongyu\;\; Wang}$ \footnotetext{The authors are with the
School of Mathematics, Yangzhou University, Jiangsu Province,
225009, China
\par J.Luo is also with the Division of Mathematics, School
of Physics and Mathematical Sciences, Nanyang Technological
University, Singapore.
\par \quad E-mail addresses: jqluo@ntu.edu.sg, ystang@yzu.edu.cn, hywang@yzu.edu.cn.}
\end{center}
\newpage
 \textbf{Abstract} \par Let $q=p^n$ with $n=2m$ and $p$ be an odd
 prime.
Let $0\leq k\leq n-1$ and $k\neq m$. In this paper we determine the
value distribution of following exponential(character) sums
\[\sum\limits_{x\in
\bF_q}\zeta_p^{\Tra_1^m (\alpha x^{p^{m}+1})+\Tra_1^n(\beta
x^{p^k+1})}\quad(\alpha\in \bF_{p^m},\beta\in \bF_{q})\]
 and
\[\sum\limits_{x\in
\bF_q}\zeta_p^{\Tra_1^m (\alpha x^{p^{m}+1})+\Tra_1^n(\beta
x^{p^k+1}+\ga x)}\quad(\alpha\in \bF_{p^m},\beta,\ga\in \bF_{q})\]

 where $\Tra_1^n: \bF_q\ra \bF_p$ and $\Tra_1^m: \bF_{p^m}\ra\bF_p$ are the canonical trace
mappings and $\zeta_p=e^{\frac{2\pi i}{p}}$ is a primitive $p$-th
root of unity. As applications:
 \begin{itemize}
    \item[(1).]  We determine the weight distribution
of the cyclic codes $\cC_1$ and $\cC_2$ over $\bF_{p^t}$ with
parity-check polynomials $h_2(x)h_3(x)$ and $h_1(x)h_2(x)h_3(x)$
respectively where $t$ is a divisor of $d=\gcd(m,k)$, and $h_1(x)$,
$h_2(x)$ and $h_3(x)$ are the minimal polynomials of $\pi^{-1}$,
$\pi^{-(p^k+1)}$ and $\pi^{-(p^m+1)}$ over $\bF_{p^t}$ respectively
for a primitive element $\pi$ of $\bF_q$.
    \item[(2).]We determine the correlation distribution among
    a family of m-sequences. \end{itemize}
 This paper
extends the results in \cite{Zen Li}.

\newpage
\section{Introduction}

\quad Let $C$ be an $[l,k,d]_{p^t}$ cyclic code and $A_i$ be the
number of codewords in $\cC$ with Hamming weight $i$. The weight
distribution $\{A_i\}_{i=0}^{l}$ is an important research object in
coding theory. If $\cC$ is irreducible, which means that the
parity-check polynomial of $\cC$ is irreducible in $\bF_{p^t}[x]$,
the weight of each codeword can be expressed by certain conbination
of Gaussian sums so that the weight distribution of $\cC$ can be
determined if the corresponding Gaussian sums can be calculated
explicitly (see Fitzgerald and Yucas \cite{Fit Yuc}, McEliece
\cite{McEl}, McEliece and Rumsey \cite{McE Rum},  van der Vlugt
\cite{Vand}, Wolfmann \cite{Wolf} and the references therein). As
for the relationship between the weight distribution of cyclic codes
and the rational points of certain curves, see Schoof \cite{Scho}.

For a general cyclic code, the Hamming weight of each codeword can
be expressed by certain combination of more general
exponential(character) sums (see Feng and Luo \cite{Fen Luo},
\cite{Fen Luo2}, Luo and Feng \cite{Luo Fen}, \cite{Luo Fen2}, van
der Vlugt \cite{Vand2}, Yuan, Carlet and Ding \cite{Yua Car}). More
exactly speaking, let $q=p^n$ with $t\mid n$, $\cC$ be the cyclic
code over $\bF_{p^t}$ with length $l=q-1$ and parity-check
polynomial
\[h(x)=h_1(x)\cdots h_u(x)\quad (u\geq 2)\]
where $h_i(x)$ $(1\leq i\leq e)$ are distinct irreducible
polynomials in $\bF_{p^t}[x]$ with degree $e_i$ $(1\leq i\leq u)$,
then $\mathrm{dim}_{\bF_{p^t}}\cC=\sum\limits_{i=1}^{u}e_i$. Let
$\pi$ be a primitive element of $\bF_q$ and $\pi^{-s_i}$ be a zero
of $h_i(x)$, $1\leq s_i\leq q-2$ $(1\leq i\leq u).$ Then the
codewords in $\cC$ can be expressed by
\[c(\alpha_1,\cdots,\alpha_u)=(c_0,c_1,\cdots,c_{l-1})\quad (\alpha_1,\cdots,\alpha_u\in \bF_q)\]
where
$c_i=\sum\limits_{\lambda=1}^{u}\Tra^n_{t}(\alpha_{\lambda}\pi^{is_{\lambda}})$
$(0\leq i\leq n-1)$ and $\Tra^{h}_j:\bF_{p^h}\ra \bF_{p^j}$ is the
trace mapping for positive integers $j\mid h$. Therefore the Hamming
weight of the codeword $c=c(\alpha_1,\cdots,\alpha_u)$ is
{\setlength\arraycolsep{2pt}
\begin{eqnarray} \label{Wei}
w_H\left(c\right)&=& \#\left\{i\left|0\leq i\leq l-1,c_i\neq
0\right.\right\}
\nonumber\\[1mm]
&=& l-\#\left\{i\left|0\leq i\leq l-1,c_i=0\right.\right\}
\nonumber\\[1mm]
&=& l-\frac{1}{p^t}\,\sum\limits_{i=0}^{l-1}\sum\limits_{a\in
\bF_{p^t}}\zeta_p^{\Tra_1^{t}\left(a\cdot\Tra_{t}^n\left(\sum\limits_{\lambda=1}^{u}\alpha_{\lambda}\pi^{is_{\lambda}}\right)\right)}
\nonumber\\[1mm]
&=&l-\frac{l}{p^t}-\frac{1}{p^t}\,\sum\limits_{a\in
\bF_{p^t}^*}\sum\limits_{x\in \bF_q^*}\zeta_p^{\Tra_1^n(af(x))}
\nonumber
\\[1mm]
&=&l-\frac{l}{p^t}+\frac{p^t-1}{p^t}-\frac{1}{p^t}\,\sum\limits_{a\in \bF_{p^t}^*}S(a\alpha_1,\cdots,a\alpha_u)\nonumber\\[1mm]
&=&p^{n-t}(p^t-1)-\frac{1}{p^t}\,\sum\limits_{a\in
\bF_{p^t}^*}S(a\alpha_1,\cdots,a\alpha_u)
\end{eqnarray}
} where
$f(x)=\alpha_1x^{s_1}+\alpha_2x^{s_2}+\cdots+\alpha_ux^{s_u}\in
\bF_{q}[x]$, $\bF_q^*=\bF_q\backslash\{0\}$,
$\bF_{p^t}^*=\bF_{p^t}\backslash\{0\}$,  and
\[S(\alpha_1,\cdots,\alpha_u)=\sum\limits_{x\in \bF_q}\zeta_p^{\Tra_1^n(\alpha_1x^{s_1}+\cdots+\alpha_ux^{s_u})}.\]
In this way, the weight distribution of cyclic code $\cC$ can be
derived from the explicit evaluating of the exponential sums
\[S(\alpha_1,\cdots,\alpha_u)\quad(\alpha_1,\cdots,\alpha_u\in \bF_q).\]

Let $n=2m, 0\leq k\leq n-1, k\neq m$, $p$ be an odd prime,
$d=\gcd(k,m)$ and $q_0=p^d$. Define $s=n/d$. Then we have
$q=q_0^{s}$. Assume $t$ is a divisor of $d$ and $n_0=n/t$. Let
$h_1(x)$, $h_2(x)$ and $h_{3}(x)$ be the minimal polynomials of
$\pi^{-1},\pi^{-(p^k+1)}$ and $\pi^{-{(p^m+1)}}$ over $\bF_{p^t}$
respectively. Then
\begin{equation}\label{deg}
\mathrm{deg}\,h_i(x)=n_0\; \text{for}\; i=1,2\;\text{and}\;
\mathrm{deg}\,h_3(x)=n_0/2
\end{equation}

 Let
$\cC_1$ and $\cC_2$ be the cyclic codes over $\bF_{p^t}$ with length
$l=q-1$ and parity-check polynomials $h_2(x)h_3(x)$ and
$h_1(x)h_2(x)h_3(x)$ respectively.  From (\ref{deg}), we know that
the dimensions of $\cC_1$ and $\cC_2$ over $\bF_{p^t}$ are $3n_0/2$
and $5n_0/2$ respectively.

For $\al\in \bF_{p^m},(\be,\ga)\in \bF_q^2$,  define the exponential
sums
\begin{equation}\label{def T}
T(\al,\be)=\sum\limits_{x\in \bF_q}\zeta_p^{\Tra_1^m (\alpha
x^{p^{m}+1})+\Tra_1^n(\beta x^{p^k+1})}
\end{equation}
and
\begin{equation}\label{def S}
S(\al,\be,\ga)=\sum\limits_{x\in \bF_q}\zeta_p^{\Tra_1^m (\alpha
x^{p^{m}+1})+\Tra_1^n(\beta x^{p^k+1}+\ga x)}.
\end{equation}

Then the weight distribution of $\cC_1$ and $\cC_2$ can be
completely determined if $T(\al,\be)$ and $S(\al,\be,\ga)$ are
explicitly evaluated.

Another application of $S(\al,\be,\ga)$ is to calculate the
correlation distribution of corresponding sequences. Let
$\mathcal{F}$ be a collection of $p$-ary m-sequences of period $q-1$
defined by

\[\mathcal{F}=\left\{\left\{a_i(t)\right\}_{i=0}^{q-2}|\,0\leq i\leq L-1 \right\}\]

The \emph{correlation function} of $a_i$ and $a_j$ for a shift
$\tau$ is defined by

\[M_{{i},{j}}(\tau)=\sum\limits_{\lambda=0}^{q-2}\zeta_p^{a_i({\lambda})-a_j({\lambda+\tau})}\hspace{2cm}(0\leq \tau\leq
q-2).\]

 In this paper, we will study the collection of sequences

 \begin{equation}\label{def F}
  \calF=\left\{a_{\al,\be}=\left\{a_{\al,\be}(\pi^{\lambda})\right\}_{\lambda=0}^{q-2}|\,\al\in \bF_{p^m}, \be\in \bF_{q} \right\}
 \end{equation}
where
  $a_{\al,\be}(\pi^{\lambda})=\Tra_1^m(\al \pi^{\lambda(p^m+1)})+\Tra_1^n(\be
  \pi^{\lambda(p^k+1)}+\pi^{\lambda})$.

Then the correlation function between $a_{\al_1,\be_1}$ and
$a_{\al_2,\be_2}$ by a shift $\tau$ ($0\leq \tau\leq q-2$) is
\begin{equation}\label{cor fun}
\begin{array}{ll}
&M_{(\al_1,\be_1),(\al_2,\be_2)}(\tau)=\sum\limits_{\lambda=0}^{q-2}\zeta_p^{a_{\al_1,\be_1}({\lambda})-
a_{\al_2,\be_2}({\lambda+\tau})}\\[2mm]
&\qquad =\sum\limits_{\lambda=0}^{q-2}\zeta_p^{\Tra_1^m(\al_1
\pi^{\lambda(p^m+1)})+\Tra_1^n(\be_1
  \pi^{\lambda(p^k+1)}+\pi^{\lambda})-\Tra_1^m(\al_2 \pi^{(\lambda+\tau)(p^m+1)})-\Tra_1^n(\be
  \pi^{(\lambda+\tau)(p^k+1)}+\pi^{\lambda+\tau})}\\[2mm]
  &\qquad = S(\al',\be',\ga')-1
  \end{array}
\end{equation}
 where
 \begin{equation}\label{coe cor}
 \al'=\al_1-\al_2 \pi^{\tau(p^m+1)},\quad
 \be'=\be_1-\be_2\pi^{\tau(p^k+1)},\quad \ga'=1-\pi^{\tau}.
 \end{equation}

Pairs of $p$-ary m-sequences with few-valued cross correlations have
been extensively studied for several decades, see Gold \cite{Gold},
Helleseth and Kumar \cite{Hel Kum}, Helleseth, Lahtonen and
Rosendahl \cite{Hel Lah Ros}, Kasami \cite{Kasa},  Rosendahl
\cite{Rose1}, \cite{Rose2} and Trachtenberg \cite{Trac}.

 Several special cases of exponential sums (\ref{def S}) and related cyclic code $\cC_2$ have been
investigated, for instance
\begin{itemize}
  \item The binary code $\cC_2$ with $k=m\pm 1$ is nothing but the
  classical Kasami code, see Kasami \cite{Kasa}.
  \item As for the binary code $\cC_2$ with $k=1$, its minimal
  distance is obtained by Lahtonen \cite{Laht}, Moreno and Kumar
  \cite{Mor Kum}. Its weight distribution is determined eventually
  in van der Vlugt \cite{Vand2}.
  \item For several other cases, the binary code $\cC_2$
  and the related family of generalized Kasami sequences have been
  studied, see Zeng, Liu and Hu \cite{Zen Liu}.
  \item In the case $p$ odd prime and $\gcd(m, k)=\gcd(m+k,2k)=d$
  being odd, the weight distribution of $\cC_2$ and correlation
  distribution of corresponding sequences have been fully
  determined, see Zeng, Li and Hu \cite{Zen Li}.
\end{itemize}

 This paper is presented as follows. In Section 2 we introduce
some preliminaries. In Section 3 we will study the value
distribution of $T(\al,\be)$ (that is, which value $T(\al,\be)$
takes on and which frequency of each value for $\al\in
\bF_{p^m},\be\in \bF_q$) and the weight distribution of $\cC_1$. In
Section 3 we will determine the value distribution of
$S(\al,\be,\ga)$ , the correlation distribution among the sequences
in $\calF$, and then the weight distribution of $\cC_2$. Most
lengthy details are presented in several appendixes. The main tools
are quadratic form theory over finite fields of odd characteristic,
some moment identities on $T(\alpha,\beta)$ and a class of
Artin-Schreier curves on finite fields. We will focus our study on
the odd prime characteristic case and the binary case will be
investigated in a following paper.

\section{Preliminaries}

\quad We follow the notations in Section 1. The first machinery to
determine the values of exponential sums $T(\alpha,\beta)$ and
$S(\al,\be,\ga)$ defined in (\ref{def T}) and (\ref{def S}) is
quadratic form theory over $\bF_{q_0}$.

 Let $H$
be an $s\times s$ symmetric matrix over $\bF_{q_0}$ and
$r=\mathrm{rank}\,H$. Then there exists $M\in
\mathrm{GL}_s(\bF_{q_0})$ such that $H'=MHM^T$ is diagonal and
$H'=diag(a_1,\cdots,a_r,0,\cdots,0)$ where $a_i\in \bF_{q_0}^*$
($1\leq i\leq r$). Let $\Delta=a_1\cdots a_r$ (we assume $\Delta=1$
when $r=0$) and $\eta_0$ be the quadratic (multiplicative) character
of $\bF_{q_0}$. Then $\eta_0(\Delta)$ is an invariant of $H$ under
the conjugate action of $M\in \mathrm{GL}_s(\bF_{q_0})$.

For the quadratic form
\begin{equation}\label{qua for}
F:\bF_{q_0}^s\ra \bF_{q_0},\quad F(x)=XHX^T\quad
(X=(x_1,\cdots,x_s)\in \bF_{q_0}^s),
\end{equation}
 we have the following result(see \cite{Luo Fen}, Lemma 1).
\begin{lemma}\label{qua}
(i). For the quadratic form $F=XHX^T$ defined in (\ref{qua for}), we
have
\[
\sum\limits_{X\in\bF_{q_0}^s}\zeta_p^{\Tra_1^{d}(F(X))} =\left\{
\begin{array}{ll}
\eta_0(\Delta)q_0^{s-r/2} & \ \ \hbox{if} \ q_0\equiv
1\;(\mathrm{mod}\;
4),\\[2mm]
i^r\eta_0(\Delta){q_0}^{s-r/2}  & \ \ \hbox{if} \ {q_0}\equiv 3\;(\mathrm{mod}\; 4).\\
\end{array}
\right.
\]
(ii). For $A=(a_1,\cdots,a_s)\in \bF_{q_0}^s$, if $2YH+A=0$ has
solution $Y=B\in \bF_{q_0}^s$,

then
$\sum\limits_{X\in\bF_{q_0}^s}\zeta_p^{\Tra_1^{d}(F(X)+AX^T)}=\zeta_p^c\sum\limits_{X\in
\bF_{q_0}^s}\zeta_p^{\Tra_1^{d}\left({F(X)}\right)}$ where
$c=-\Tra_1^{d}\left(BHB^T\right)=\frac{1}{2}\Tra_1^{d}\left(AB^T\right)\in
\bF_p$.

Otherwise
$\sum\limits_{X\in\bF_p^m}\zeta_p^{\Tra_1^{d}(F(X)+AX^T)}=0$.
\end{lemma}

In this correspondence we always assume $d=\gcd(m,k)$. Recall that
$s=n/d$ is even. Therefore the field $\bF_q$ is a vector space over
$\bF_{q_0}$ with dimension $s$. We fix a basis $v_1,\cdots,v_s$ of
$\bF_q$ over $\bF_{q_0}$. Then each $x\in \bF_q$ can be uniquely
expressed as
\[x=x_1v_1+\cdots+x_sv_s\quad (x_i\in \bF_{q_0}).\]
Thus we have the following $\bF_{q_0}$-linear isomorphism:
\[\bF_q\xrightarrow{\sim}\bF_{q_0}^s,\quad x=x_1v_1+\cdots+x_sv_s\mapsto
X=(x_1,\cdots,x_s).\] With this isomorphism, a function $f:\bF_q\ra
\bF_{q_0}$ induces a function $F:\bF_{q_0}^s\ra \bF_{q_0}$ where for
$X=(x_1,\cdots,x_s)\in \bF_{q_0}^s, F(X)=f(x)$ with
$x=x_1v_1+\cdots+x_sv_s$. In this way, function
$f(x)=\Tra_{d}^n(\gamma x)$ for $\gamma\in \bF_q$ induces a linear
form \begin{equation} F(X)=\Tra_{d}^n(\gamma
x)=\sum\limits_{i=1}^{s}\Tra_{d}^n(\gamma v_i)x_i=A_{\ga}X^T
\end{equation}\label{def A_gamma}
 where $A_{\ga}=\left(\Tra_{d}^n(\gamma
v_1),\cdots,\Tra_{d}^n(\gamma v_s)\right),$
 and
$f_{\alpha,\beta}(x)=\Tra_{d}^m(\alpha x^{p^m+1})+\Tra_d^n(\beta
x^{p^k+1})$  for $\al\in \bF_{p^m}$, $\be\in \bF_q$ induces a
quadratic form

\begin{eqnarray}\label{def H_al be}
F_{\alpha,\beta}(X)&=&\Tra_{d}^m(\alpha x^{p^m+1})+\Tra_d^n(\beta
x^{p^k+1})\nonumber\\[5mm]
&=&\Tra_{d}^m\left(\alpha\left(\sum\limits_{i=1}^s
x_iv_i^{p^m}\right)\left(\sum\limits_{i=1}^s
x_iv_i\right)\right)+\Tra_{d}^n\left(\left(\beta\sum\limits_{i=1}^s
x_iv_i^{p^k}\right)\left(\sum\limits_{i=1}^s
x_iv_i\right)\right)\nonumber\\[2mm]
&=&\sum\limits_{i,j=1}^s\left(\frac{1}{2}\Tra_{d}^m\left(\alpha
v_i^{p^m}v_j+\alpha v_iv_j^{p^m}\right) +\Tra_d^n\left(\be
v_i^{p^k}v_j\right)\right)x_ix_j=XH_{\alpha,\beta}X^T
\end{eqnarray}

 where \[H_{\alpha,\beta}=(h_{ij})\;\hbox{and}\;
h_{ij}=\frac{1}{2}\Tra_{d}^m\left(\alpha v_i^{p^m}v_j+\alpha
v_iv_j^{p^m}\right)+\frac{1}{2}\Tra_{d}^n\left(\be v_i^{p^k}v_j+\be
v_iv_j^{p^k}\right)\;\hbox{for}\;1\leq i,j\leq s.\]

From Lemma \ref{qua},  in order to determine the values of
\[T(\alpha,\beta)=\sum\limits_{x\in \bF_q}\zeta_p^{\Tra_1^m (\alpha
x^{p^{m}+1})+\Tra_1^n(\beta x^{p^k+1})}=\sum\limits_{X\in
\bF_{q_0}^s}\zeta_p^{\Tra_1^{d}\left(XH_{\alpha,\beta}X^T\right)}\]
and
\[S(\alpha,\beta,\ga)=\sum\limits_{x\in \bF_q}\zeta_p^{\Tra_1^m (\alpha
x^{p^{m}+1})+\Tra_1^n(\beta x^{p^k+1}+\ga x)}=\sum\limits_{X\in
\bF_p^m}\zeta_p^{\Tra_1^{d}\left(XH_{\alpha,\beta}X^T+A_{\ga}X^T\right)}\quad
(\alpha\in \bF_{p^m},\beta,\ga\in \bF_q),\] we need to determine the
rank of $H_{\alpha,\beta}$ over $\bF_{q_0}$ and the solvability of
$\bF_{q_0}$-linear equation $2XH_{\al,\be}+A_{\ga}=0$.

Define $d'=\gcd(m+k,2k)$. Then an easy observation shows
\begin{equation}\label{rel d d'}
d'=\left\{
\begin{array}{ll}
2d, & \text{if}\; m/d\;\text{and}\; k/d \;\text{are both odd};\\[1mm]
d, &\text{otherwise.}
\end{array}
\right.
\end{equation}

The main part of the subsequent result has been proven in \cite{Zen
Li} and we repeat part of the proof for self-containing.
\begin{lemma}\label{rank}
For $(\alpha,\beta)\in \bF_{p^m}\times \bF_q\backslash\{(0,0)\}$,
let $r_{\alpha,\beta}$ be the rank of $H_{\alpha,\beta}$.  Then we
have
\begin{itemize}
  \item[(i).] if $d'=d$, then the possible values of $r_{\al,\be}$
  are $s$, $s-1$, $s-2$.
  \item[(ii).] if $d'=2d$, then the possible values of $r_{\al,\be}$
  are $s$, $s-2$, $s-4$.
\end{itemize}

Moreover, let $n_i$ be the number of $({\alpha,\beta})$ with
$r_{\alpha,\beta}=s-i$. In the case $d'=d$, we have
$n_1=p^{m-d}(p^n-1)$.
\end{lemma}
\begin{proof}
see \textbf{Appendix A.}
\end{proof}


 In order to determine the value distribution of
$T(\al,\be)$
 for $\al\in \bF_{p^m},\be\in \bF_q$, we need the
following result on moments of $T(\al,\be)$.
\begin{lemma}\label{moment}
For the exponential sum $T(\al,\be)$,
\[\begin{array}{ll}&(i). \;\;\sum\limits_{\al\in \bF_{p^m},\be\in
\bF_q}T(\al,\be)=p^{3m};\\[2mm]
                   &(ii).
                   \sum\limits_{\al\in \bF_{p^m},\be\in
\bF_q}T(\al,\be)^2=\left\{\begin{array}{ll}
p^{3m} &\text{if}\; d'=d\;\text{and}\;\; p^d\equiv 3\pmod 4,\\[1mm]
(2p^n-1)\cdot
p^{3m}&\text{if}\; d'=d\;\text{and}\; p^d\equiv 1\pmod 4,\\[1mm]
(p^{n+d}+p^n-p^d)\cdot p^{3m}&\text{if}\; d'=2d;
\end{array}\right.\\[2mm]
                   &(iii). \;\;\text{if}\;d'=d,\;\text{then}\\[2mm]
                   &\quad\quad\quad \sum\limits_{\al\in \bF_{p^m},\be\in
\bF_q}T(\al,\be)^3=(p^{n+d}+p^{n}-p^d)\cdot p^{3m}.\\[4mm]
\end{array}\]
\end{lemma}
\begin{proof}
see  \textbf{Appendix A.}
\end{proof}

In the case $d'=2d$, we could determine the explicit values of
$T(\al,\be)$. To this end we will study a class of Artin-Schreier
curves. A similar technique has been applied in Coulter \cite{Coul},
Theorem 6.1.

\begin{lemma}\label{Artin}Suppose $(\al,\be)\in (\bF_{p^m}\times
\bF_q)\big{\backslash}\{0,0\}$ and $d'=2d$. Let $N$ be the number of
$\bF_q$-rational (affine) points on the curve
\begin{equation}\label{Artin Sch}
\frac{1}{2}\al x^{p^m+1}+\be x^{p^k+1}=y^{p^d}-y.
\end{equation}
 Then
\[N=q+(p^d-1)\cdot T(\al,\be).\]
\end{lemma}
\begin{proof}
see \textbf{Appendix A.}
\end{proof}

Now we give an explicit evaluation of $T(\al,\be)$ in the case
$d'=2d$.
\begin{lemma}\label{reduce num}
Assumptions as in Lemma \ref{Artin}, then
\[
T(\al,\be)=\left\{
\begin{array}{ll}
-p^m, &\text{if}\; r_{\al,\be}=s\\[2mm]
p^{m+d}, &\text{if}\; r_{\al,\be}=s-2\\[2mm]
-p^{m+2d}, &\text{if}\; r_{\al,\be}=s-4
\end{array}
\right.
\]
\end{lemma}
\begin{proof}
Consider the $\bF_q$-rational (affine) points on the Artin-Schreier
curve in Lemma \ref{Artin}. It is easy to verify that $(0,y)$ with
$y\in \bF_{p^d}$ are exactly the points on the curve with $x=0$. If
$(x,y)$ with $x\neq 0$ is a point on this curve, then so are $(t
x,t^{p^d+1}y)$ with $t^{p^{2d}-1}=1$ (note that $p^m+1\equiv
p^k+1\equiv p^d+1\pmod {p^{2d}-1}$ since $m/d$ and $k/d$ are both
odd by (\ref{rel d d'})). In total, we have
\[q+(p^d-1)T(\al,\be)=N\equiv p^d\pmod {p^{2d}-1}\]
which yields
\[T(\al,\be)\equiv 1\pmod{p^d+1}.\]

 We only consider the case $r_{\al,\be}=s$. The other cases are similar. In this case $T(\al,\be)=\pm
 p^m$. Assume $T(\al,\be)=p^m$.
 Then ${p^d+1}\mid p^m-1$ which contradicts to $m/d$ is odd. Therefore $T(\al,\be)=-p^m$.
\end{proof}
\begin{remark}
(i). Our treatment improve the technique in \cite{Coul}. Otherwise
the case $(p,d)=(3,1)$ will be excluded.

(ii). Applying Lemma \ref{reduce num} to Lemma \ref{Artin}, we could
determine the number of rational points on the curve (\ref{Artin
Sch}).
\end{remark}

\section{Exponential Sums $T(\al,\be)$ and Cyclic Code $\cC_1$}

\quad Recall $q_0^*=(-1)^{\frac{q_0-1}{2}}q_0$.
 In this section we prove the following results.

\begin{theo}\label{value dis T}
The value distribution of the multi-set
$\left\{T(\al,\be)\left|\al,\be\in \bF_q\right.\right\}$ is shown as
following.

 (i). For the case $d'=d$,
\begin{center}
\begin{tabular}{|c|c|}
\hline
values & multiplicity \\[2mm]
\hline $p^{m}$&$p^d(p^m-1)(p^m+1)^2\big{/}\left(2(p^d+1)\right)$
\\[2mm]
\hline
$-p^{m}$&$p^d(p^m-1)(p^{n}-2p^{n-d}+1)\big{/}\left({2(p^d-1)}\right)$
\\[2mm]
\hline

$\sqrt{{q_0}^*}{q_0}^{\frac{s}{2}},
-\sqrt{{q_0}^*}{q_0}^{\frac{s}{2}}$&$\frac{1}{2}p^{m-d}(p^{n}-1)$
\\[2mm]
\hline
$-p^{m+d}$&$(p^{m-d}-1)(p^{n}-1)\big{/}\left(p^{2d}-1\right)$\\[2mm]
\hline $p^n$&$1$
\\[2mm]
\hline
\end{tabular}
\end{center}

(ii). For the case $d'=2d$,
\begin{center}
\begin{tabular}{|c|c|}
\hline
values & multiplicity \\[2mm]
\hline
$-p^m$&$p^{3d}(p^m-1)(p^{n}-p^{n-2d}-p^{n-3d}+p^m-p^{m-d}+1)\big{/}\left((p^d+1)(p^{2d}-1)\right)$
\\[2mm]
\hline
${p}^{m+d}$&$p^{d}(p^n-1)(p^m+p^{m-d}+p^{m-2d}+1)\big{/}(p^d+1)^2$
\\[2mm]
\hline
 $-{p}^{m+2d}$&$
(p^{m-d}-1)(p^{n}-1)\big{/}\left((p^d+1)(p^{2d}-1)\right)$
\\[2mm]

\hline $p^n$&$1$
\\[2mm]
\hline
\end{tabular}
\end{center}
\end{theo}

\begin{proof}
see \textbf{Appendix B.}
\end{proof}
Recall that $t$ is a divisor of $d$ and $\cC_1$ is the cyclic code
over $\bF_{p^t}$ with parity-check polynomial $h_2(x)h_3(x)$ where
$h_2(x)$ and $h_3(x)$ are the minimal polynomials of
$\pi^{-(p^k+1)}$ and $\pi^{-(p^m+1)}$, respectively.
\begin{theo}\label{wei dis C1}
For $k\neq m$, the weight distribution $\{A_0,A_1,\cdots,A_l\}$ of
the cyclic code $\cC_1$ over $\bF_{p^t}$ ($p\geq 3$) with length
$l=q-1$ and $\mathrm{dim}_{\bF_{p^t}}\cC_1=3n_0/2$ is shown as
following.

(i). For the case $d'=d$ and $d/t$ is odd,
\begin{center}
\begin{tabular}{|c|c|}
\hline
$i$ & $A_i$ \\[2mm]
\hline
$(p^t-1)(p^{n-t}-p^{m-t})$&$p^d(p^m-1)(p^m+1)^2\big{/}\left(2(p^d+1)\right)$
\\[2mm]
\hline

$(p^t-1)p^{n-t}$&$p^{m-d}(p^{n}-1)$
\\[2mm]
\hline
$(p^t-1)(p^{n-t}+p^{m-t})$&$p^d(p^m-1)(p^{n}-2p^{n-d}+1)\big{/}\left({2(p^d-1)}\right)$
\\[2mm]
\hline
$(p^t-1)(p^{n-t}+p^{m+d-t})$&$(p^{m-d}-1)(p^{n}-1)\big{/}\left(p^{2d}-1\right)$\\[2mm]
\hline $0$&$1$
\\[2mm]
\hline
\end{tabular}
\end{center}

(ii). For the case $d'=d$ and $d/t$ is even,
\begin{center}
\begin{tabular}{|c|c|}
\hline
$i$ & $A_i$ \\[2mm]
\hline

$(p^t-1)(p^{n-t}-p^{m+\frac{d}{2}-t})$&$\frac{1}{2}p^{m-d}(p^{n}-1)$
\\[2mm]
\hline
$(p^t-1)(p^{n-t}-p^{m-t})$&$p^d(p^m-1)(p^m+1)^2\big{/}\left(2(p^d+1)\right)$
\\[2mm]
\hline
$(p^t-1)(p^{n-t}+p^{m-t})$&$p^d(p^m-1)(p^{n}-2p^{n-d}+1)\big{/}\left({2(p^d-1)}\right)$
\\[2mm]
\hline
$(p^t-1)(p^{n-t}+p^{m+\frac{d}{2}-t})$&$\frac{1}{2}p^{m-d}(p^{n}-1)$
\\[2mm]
\hline
$(p^t-1)(p^{n-t}+p^{m+d-t})$&$(p^{m-d}-1)(p^{n}-1)\big{/}\left(p^{2d}-1\right)$\\[2mm]
\hline $0$&$1$
\\[2mm]
\hline
\end{tabular}
\end{center}

(iii). For the case $d'=2d$,
\begin{center}
\begin{tabular}{|c|c|}
\hline
$i$ & $A_i$ \\[2mm]
\hline
$(p^t-1)(p^{n-t}-p^{m+d-t})$&$p^{d}(p^n-1)(p^m+p^{m-d}+p^{m-2d}+1)\big{/}(p^d+1)^2$
\\[2mm]
\hline
$(p^t-1)(p^{n-t}+p^{m-t})$&$p^{3d}(p^m-1)(p^{n}-p^{n-2d}-p^{n-3d}+p^m-p^{m-d}+1)\big{/}\left((p^d+1)(p^{2d}-1)\right)$
\\[2mm]
\hline

 $(p^t-1)(p^{n-t}+p^{m+2d-t})$&$
(p^{m-d}-1)(p^{n}-1)\big{/}\left((p^d+1)(p^{2d}-1)\right)$
\\[2mm]

\hline $0$&$1$
\\[2mm]
\hline
\end{tabular}
\end{center}
\end{theo}
\begin{proof}
see \textbf{Appendix B.}
\end{proof}
\begin{remark}
\begin{itemize}
    \item[(1).] In the case $d=d'$. Since $\gcd(p^m+1,
p^k+1)=2$, the first $l'=\frac{q-1}{2}$ coordinates of each codeword
of $\cC_1$ form a cyclic code $\cC_1'$ over $\bF_{p^t}$ with length
$l'$ and dimension $3n_0/2$. Let $(A_0',\cdots,A_{l'}')$ be the
weight distribution of $\cC_1'$, then $A_i'=A_{2i}$ $(0\leq i\leq
l')$.
    \item[(2).] In the case $d'=2d$. Since $\gcd(p^m+1,
p^k+1)=p^d+1$, the first $l'=\frac{q-1}{p^d+1}$ coordinates of each
codeword of $\cC_1$ form a cyclic code $\cC_1'$ over $\bF_{p^t}$
with length $l'$ and dimension $3n_0/2$. Let $(A_0',\cdots,A_{l'}')$
be the weight distribution of $\cC_1'$, then $A_i'=A_{(p^d+1)i}$
$(0\leq i\leq l')$.
\item[(2).] If $k=0$, this result is the same as \cite{Luo Fen},
    Theorem 3.
\end{itemize}

\end{remark}

\section{Results on Correlation Distribution of Sequences and Cyclic Code $\cC_2$}

\quad Recall $\phi_{\al,\be}(x)$ in the proof of Lemma \ref{rank}
and $N_{i,\veps}$ in the proof of Theorem \ref{value dis T}. Finally
we will determine the value distribution of $S(\al,\be,\ga)$, the
correlation distribution among sequences in $\calF$ defined in
(\ref{def F}) and the weight distribution of $\cC_2$ defined in
Section 1. The following result will play an important role.
\begin{lemma}\label{last lemma}
Let $t$ be a divisor of $d$. For any $a\in \bF_{p^t}$ and any
$(\al,\be)\in N_{i,\veps}$ with $\veps=\pm 1$ and $0\leq i\leq 4$,
then the number of elements $\ga\in \bF_q$ satisfying
\begin{itemize}
    \item[(i).] $\phi_{\al,\be}(x)+\ga=0$ is solvable(choose one solution, say
    $x_0$),
    \item[(ii).]  $\Tra_t^m (\alpha
x_0^{p^{m}+1})+\Tra_t^n(\beta x_0^{p^k+1})=a$
\end{itemize}
is
\[
\left\{
    \begin{array}{ll}
    p^{n-id-t}&\text{if}\; s-i\; \text{and}\; d/t\;\text{are both
    odd},\;\text{and}\; a=0,\\[2mm]
    p^{n-id-t}+\veps \eta'(a)p^{\frac{n-id-t}{2}}&\text{if}\; s-i\; \text{and}\; d/t\;\text{are both
    odd},\;\text{and}\; a\neq 0,\\[2mm]
    p^{n-id-t}+
    \veps (p^t-1)p^{\frac{n-id}{2}-t}&\text{if}\; s-i\; \text{or}\; d/t\;\text{is
    even},\;\text{and}\; a=0,\\[2mm]
    p^{n-id-t}-\veps p^{\frac{n-id}{2}-t}&\text{if}\; s-i\; \text{or}\; d/t\;\text{is
    even},\;\text{and}\; a\neq 0.
    \end{array}
\right.
\]
where $\eta'$ is the quadratic (multiplicative) character on
$\bF_{p^t}$.
\end{lemma}
\begin{proof}
see \textbf{Appendix C.}
\end{proof}

Let $p^*=(-1)^{\frac{p-1}{2}}p$ and $\left(\frac{\cdot}{p}\right)$
be the Legendre symbol.
 We are now ready to give the value
distribution of $S(\al,\be,\ga)$.

\begin{theo}\label{value dis S}
The value distribution of the multi-set
$\left\{S(\al,\be,\ga)\left|\al\in \bF_{p^m},(\be,\ga)\in
\bF_q^2\right.\right\}$ is shown as following.

 (i). If $d'=d$ is odd, then
\begin{center}
\begin{tabular}{|c|c|}
\hline
values & multiplicity \\[2mm]
\hline $
p^{m}$&$p^{m+d-1}(p^m+1)(p^{m}+p-1)(p^n-1)\big{/}\left(2(p^d+1)\right)$
\\[2mm]
\hline $
-p^{m}$&$p^{m+d-1}(p^m-1)(p^{m}-p+1)(p^n-2p^{n-d}+1)\big{/}\left(2(p^d-1)\right)$
\\[2mm]
\hline $ \zeta_p^j
p^{m}$&$p^{m+d-1}(p^n-1)^2\big{/}\left(2(p^d+1)\right)$
\\[2mm]
\hline $ -\zeta_p^j
p^{m}$&$p^{m+d-1}(p^n-1)(p^n-2p^{n-d}+1)\big{/}\left(2(p^d-1)\right)$
\\[2mm]

\hline
 $\veps\sqrt{{p}^*}{p}^{m+\frac{d-1}{2}}$&$\frac{1}{2}p^{3m-2d-1}(p^{n}-1)$
\\[2mm]
\hline
 $\veps\zeta_p^j\sqrt{{p}^*}{p}^{m+\frac{d-1}{2}}$&$\frac{1}{2}p^{n-\frac{3d+1}{2}}\left(p^{m-\frac{d+1}{2}}+\veps
\left(\frac{-j}{p}\right)\right)(p^{n}-1)$
\\[2mm]
\hline
$-p^{m+d}$&$p^{m-d-1}(p^{m-d}-1)(p^{n}-1)(p^{m-d}-p+1)\big{/}\left(p^{2d}-1\right)$\\[2mm]
\hline $- \zeta_p^j p^{m+d}$&$p^{m-d-1}(p^{n-2d}-1)(p^{n}-1)\big{/}\left(p^{2d}-1\right)$\\[2mm]
\hline $0$&$(p^{n}-1)(p^{3m-d}-p^{3m-2d}+p^{3m-3d}-p^{n-2d}+1)$
\\[2mm]
\hline $p^n$&$1$
\\[2mm]
\hline
\end{tabular}
\end{center}
where $\veps =\pm 1, 1\leq j\leq p-1$.

(ii). If $d'=d$ is even, then

\begin{center}
\begin{tabular}{|c|c|}
\hline
values & multiplicity \\[2mm]
\hline $
p^{m}$&$p^{m+d-1}(p^m+1)(p^{m}+p-1)(p^n-1)\big{/}\left(2(p^d+1)\right)$
\\[2mm]
\hline $
-p^{m}$&$p^{m+d-1}(p^m-1)(p^{m}-p+1)(p^n-2p^{n-d}+1)\big{/}\left(2(p^d-1)\right)$
\\[2mm]
\hline $ \zeta_p^j
p^{m}$&$p^{m+d-1}(p^n-1)^2\big{/}\left(2(p^d+1)\right)$
\\[2mm]
\hline $ -\zeta_p^j
p^{m}$&$p^{m+d-1}(p^n-1)(p^n-2p^{n-d}+1)\big{/}\left(2(p^d-1)\right)$
\\[2mm]

\hline
 $\veps{p}^{m+\frac{d}{2}}$&$\frac{1}{2}p^{n-\frac{3d}{2}-1}(p^{m-\frac{d}{2}}+\veps(p-1))(p^{n}-1)$
\\[2mm]
\hline
 $\veps\zeta_p^j{p}^{m+\frac{d}{2}}$&$\frac{1}{2}p^{n-\frac{3d}{2}-1}(p^{m-\frac{d}{2}}-\veps)(p^{n}-1)$
\\[2mm]
\hline
$-p^{m+d}$&$p^{m-d-1}(p^{m-d}-1)(p^{n}-1)(p^{m-d}-p+1)\big{/}\left(p^{2d}-1\right)$\\[2mm]
\hline
$- \zeta_p^j p^{m+d}
$&$p^{m-d-1}(p^{m-d}-1)(p^{n}-1)(p^{m-d}+1)\big{/}\left(p^{2d}-1\right)$\\[2mm]
\hline $0$&$(p^{n}-1)(p^{3m-d}-p^{3m-2d}+p^{3m-3d}-p^{n-2d}+1)$
\\[2mm]
\hline $p^n$&$1$
\\[2mm]
\hline
\end{tabular}
\end{center}
 where $\veps =\pm 1, 1\leq j\leq p-1$.

(iii). If $d'=2d$, then
\begin{center}
\begin{tabular}{|c|c|}
\hline
values & multiplicity \\[2mm]
\hline
$-p^m$&$p^{m+3d-1}(p^m-1)(p^{m}-p+1)(p^{n}-p^{n-2d}-p^{n-3d}+p^m-p^{m-d}+1)\big{/}(p^d+1)(p^{2d}-1)$
\\[2mm]
\hline $-\zeta_p^j
p^{m}$&$p^{m+3d-1}(p^{n}-1)(p^{n}-p^{n-2d}-p^{n-3d}+p^m-p^{m-d}+1)\big{/}(p^d+1)(p^{2d}-1)$
\\[2mm]
\hline
${p}^{m+d}$&$p^{m-1}(p^n-1)(p^{m-d}+p-1)(p^m+p^{m-d}+p^{m-2d}+1)\big{/}(p^d+1)^2$
\\[2mm]
\hline
$\zeta_p^j{p}^{m+d}$&$p^{m-1}(p^n-1)(p^{m-d}-1)(p^m+p^{m-d}+p^{m-2d}+1)\big{/}(p^d+1)^2$
\\[2mm]
\hline
 $-{p}^{m+2d}$&$
p^{m-2d-1}(p^{m-d}-1)(p^{m-2d}-p+1)(p^{n}-1)\big{/}\left((p^d+1)(p^{2d}-1)\right)$
\\[2mm]
\hline
 $-\zeta_p^j{p}^{m+2d}$&$
p^{m-2d-1}(p^{m-d}-1)(p^{m-2d}+1)(p^{n}-1)\big{/}\left((p^d+1)(p^{2d}-1)\right)$
\\[2mm]

\hline $0$&$
\begin{array}{ll}
&(p^{n}-1)(p^{3m-d}-p^{3m-2d}+p^{3m-3d}-p^{3m-4d}+p^{3m-5d}\\[1mm]
&\qquad+p^{n-d}-2p^{n-2d}+p^{n-3d}-p^{n-4d}+1)
\end{array}
$
\\[2mm]
\hline
 $p^n$&$1$
\\[2mm]
\hline
\end{tabular}
\end{center}
where $1\leq j\leq p-1$.

\end{theo}
\begin{proof}
see \textbf{Appendix C.}
\end{proof}

\begin{remark}
Case (i) is exactly Proposition 6 in \cite{Zen Li}.
\end{remark}

In order to give the correlation distribution among the sequences in
$\calF$, we need the following lemma (see \cite{Zen Li}, Lemma 5).

\begin{lemma}\label{q-2}
For any given $\ga\in \bF_q^*$, when $(\al,\be)$ runs through
$\bF_{p^m}\times \bF_q$, the distribution of $S(\al,\be,\ga)$ is the
same as $S(\al,\be,1)$.
\end{lemma}

As a consequence of Theorem \ref{value dis T}, Theorem \ref{value
dis S} and Lemma \ref{q-2}, we could give the correlation
distribution amidst the sequences in $\calF$.
\begin{theo}\label{cor dis}
The collection $\calF$ defined in (\ref{def F}) is a family of
$p^{3m}$ $p$-ary sequences with period $q-1$. Its correlation
distribution is given as follows.

(i). If $d'=d$ is odd, then
\begin{center}
\begin{tabular}{|c|c|}
\hline
values & multiplicity \\[2mm]
\hline $
p^{m}-1$&$p^{3m+d}(p^m+1)\left(p^{m-1}(p^{m}+p-1)(p^n-2)+1\right)\big{/}\left(2(p^d+1)\right)$
\\[2mm]
\hline $
-p^{m}-1$&$p^{3m+d}\left(p^{m-1}(p^{m}-p+1)(p^n-2)+1\right)(p^n-2p^{n-d}+1)\big{/}\left(2(p^d-1)(p^m+1)\right)$
\\[2mm]
\hline $ \zeta_p^j
p^{m}-1$&$p^{2n+d-1}(p^n-2)(p^n-1)\big{/}\left(2(p^d+1)\right)$
\\[2mm]
\hline $ -\zeta_p^j
p^{m}-1$&$p^{2n+d-1}(p^n-2)(p^n-2p^{n-d}+1)\big{/}\left(2(p^d-1)\right)$
\\[2mm]

\hline
 $\veps\sqrt{{p}^*}{p}^{m+\frac{d-1}{2}}-1$&$\frac{1}{2}p^{2n-d}\left(p^{n-d-1}(p^{n}-2)+1\right)$
\\[2mm]
\hline
 $\veps\zeta_p^j\sqrt{{p}^*}{p}^{m+\frac{d-1}{2}}-1$&$\frac{1}{2}p^{5m-\frac{3d+1}{2}}\left(p^{m-\frac{d+1}{2}}+\veps
\left(\frac{-j}{p}\right)\right)(p^{n}-2)$
\\[2mm]
\hline
$-p^{\frac{m}{2}+d}-1$&$p^{3m}(p^{m-d}-1)\left(p^{m-d-1}(p^{n}-2)(p^{m-d}-p+1)+1\right)\big{/}\left(p^{2d}-1\right)$\\[2mm]
\hline $- \zeta_p^j p^{m+d}-1$&$p^{2n-d-1}(p^{m-d}-1)(p^{n}-2)(p^{m-d}+1)\big{/}\left(p^{2d}-1\right)$\\[2mm]
\hline
$-1$&$p^{3m}(p^{n}-2)(p^{3m-d}-p^{3m-2d}+p^{3m-3d}-p^{n-2d}+1)$
\\[2mm]
\hline $p^n-1$&$p^{3m}$
\\[2mm]
\hline
\end{tabular}
\end{center}
where $1\leq j\leq p-1$ and $\veps=\pm 1$.

 (ii). If $d'=d$ is even, then

\begin{center}
\begin{tabular}{|c|c|}
\hline
values & multiplicity \\[2mm]
\hline $
p^{m}-1$&$p^{3m+d}(p^m+1)\left(p^{m-1}(p^{m}+p-1)(p^n-1)+1\right)\big{/}\left(2(p^d+1)\right)$
\\[2mm]
\hline $
-p^{m}-1$&$p^{3m+d}\left(p^{m-1}(p^{m}-p+1)(p^n-2)+1\right)(p^n-2p^{n-d}+1)\big{/}\left(2(p^d-1)(p^m+1)\right)$
\\[2mm]
\hline $ \zeta_p^j
p^{m}-1$&$p^{2n+d-1}(p^n-2)(p^n-1)\big{/}\left(2(p^d+1)\right)$
\\[2mm]
\hline $ -\zeta_p^j
p^{m}-1$&$p^{2n+d-1}(p^n-2)(p^n-2p^{n-d}+1)\big{/}\left(2(p^d-1)\right)$
\\[2mm]

\hline
 $\veps{p}^{m+\frac{d}{2}}$&$\frac{1}{2}p^{2n-d}\left(p^{m-\frac{d}{2}-1}(p^{m-\frac{d}{2}}+\veps(p-1))(p^{n}-2)+1\right)$
\\[2mm]
\hline
 $\veps\zeta_p^j{p}^{m+\frac{d}{2}}$&$\frac{1}{2}p^{5m-\frac{3d}{2}-1}(p^{m-\frac{d}{2}}-\veps)(p^{n}-2)$
\\[2mm]
\hline
$-p^{\frac{m}{2}+d}-1$&$p^{3m}(p^{m-d}-1)\left(p^{m-d-1}(p^{n}-2)(p^{m-d}-p+1)+1\right)\big{/}\left(p^{2d}-1\right)$\\[2mm]
\hline $- \zeta_p^j p^{m+d}-1$&$p^{2n-d-1}(p^{n-2d}-1)(p^{n}-2)\big{/}\left(p^{2d}-1\right)$\\[2mm]
\hline
$-1$&$p^{3m}(p^{n}-2)(p^{3m-d}-p^{3m-2d}+p^{3m-3d}-p^{n-2d}+1)$
\\[2mm]
\hline $p^n-1$&$p^{3m}$
\\[2mm]
\hline
\end{tabular}
\end{center}
where $1\leq j\leq p-1$ and $\veps=\pm 1$.

 (iii). If $d'=2d$, then
\begin{center}
\begin{tabular}{|c|c|}
\hline
values & multiplicity \\[2mm]
\hline
$-p^m-1$&$
\begin{array}{ll}
&p^{3m+3d}(p^{m-1}(p^m-p+1)(p^n-2)+1)(p^{n}-p^{n-2d}-p^{n-3d}+\\[1mm]
&\qquad p^m-p^{m-d}+1)\big{/}(p^d+1)(p^{2d}-1)(p^m+1)
\end{array}
$
\\[2mm]
\hline $-\zeta_p^j
p^{m}-1$&$p^{2n+3d-1}(p^{n}-2)(p^{n}-p^{n-2d}-p^{n-3d}+p^m-p^{m-d}+1)\big{/}(p^d+1)(p^{2d}-1)$
\\[2mm]
\hline
${p}^{m+d}-1$&$p^{3m+d}\left(p^{m-d-1}(p^{m-d}+p-1)(p^n-2)+1\right)(p^m+p^{m-d}+p^{m-2d}+1)\big{/}(p^d+1)^2$
\\[2mm]
\hline
$\zeta_p^j{p}^{m+d}-1$&$p^{2n-1}(p^n-2)(p^{m-d}-1)(p^m+p^{m-d}+p^{m-2d}+1)\big{/}(p^d+1)^2$
\\[2mm]
\hline
 $-{p}^{m+2d}-1$&$p^{3m}
(p^{m-d}-1)\left(p^{m-2d-1}(p^{m-2d}-p+1)(p^{n}-2)+1\right)\big{/}\left((p^d+1)(p^{2d}-1)\right)$
\\[2mm]
\hline
 $-\zeta_p^j{p}^{m+2d}-1$&$
p^{2n-2d-1}(p^{m-d}-1)(p^{m-2d}+1)(p^{n}-2)\big{/}\left((p^d+1)(p^{2d}-1)\right)$
\\[2mm]

\hline $-1$&$
\begin{array}{ll}
&p^{3m}(p^{n}-2)(p^{3m-d}-p^{3m-2d}+p^{3m-3d}-p^{3m-4d}+p^{3m-5d}\\[1mm]
&\qquad+p^{n-d}-2p^{n-2d}+p^{n-3d}-p^{n-4d}+1)
\end{array}
$
\\[2mm]
\hline
 $p^n-1$&$p^{3m}$
\\[2mm]
\hline
\end{tabular}
\end{center}
where $1\leq j\leq p-1$.
\end{theo}
\begin{proof}
see \textbf{Appendix C.}
\end{proof}
\begin{remark}
The case (i) has been shown in \cite{Zen Li}, Prop. 6.
\end{remark}

Recall that $\cC_2$ is the cyclic code over $\bF_{p^t}$ with
parity-check polynomial $h_1(x)h_2(x)h_3(x)$ where $h_1(x)$,
$h_2(x)$ and $h_3(x)$ are the minimal polynomials of $\pi^{-1}$,
$\pi^{-(p^k+1)}$ and $\pi^{-(p^m+1)}$ respectively. Here we are
ready to determine the weight distribution of $\cC_2$.

\begin{theo}\label{wei dis C2}
The weight distribution $\{A_0,A_1,\cdots,A_{q-1}\}$ of the cyclic
code $\cC_2$ over $\bF_{p^t}$ ($p\geq 3$) with length $q-1$ and
$\mathrm{dim}_{\bF_{p^t}}\cC_1=\frac{5}{2}n_0$ is shown as
following.

(i). If $d'=d$ and $d/t$ is odd, then
\begin{center}
\begin{tabular}{|c|c|}
\hline
$i$ & $A_i$ \\[2mm]
\hline
$(p^t-1)(p^{n-t}-p^{m-t})$&$p^{m+d-t}(p^{m}+p^t-1)(p^m-1)(p^m+1)^2\big{/}\left(2(p^d+1)\right)$
\\[2mm]
\hline
$(p^t-1)(p^{n-t}+p^{m-t})$&$p^{m+d-t}(p^{m}-p^t+1)(p^m-1)(p^{n}-2p^{n-d}+1)\big{/}\left({2(p^d-1)}\right)$
\\[2mm]
\hline

$(p^t-1)p^{n-t}+p^{m-t}$&$p^{m+d-t}(p^t-1)(p^n-1)^2\big{/}\left(2(p^d+1)\right)$
\\[2mm]
\hline

$(p^t-1)p^{n-t}-p^{m-t}$&$p^{m+d-t}(p^t-1)(p^n-1)(p^{n}-2p^{n-d}+1)\big{/}\left({2(p^d-1)}\right)$
\\[2mm]
\hline

 $(p^t-1)p^{n-t}-p^{m+\frac{d-t}{2}}$&$\frac{1}{2}p^{m-d}(p^t-1)(p^{n-d-t}+p^{\frac{n-d-t}{2}})(p^{n}-1)$
\\[2mm]
\hline

 $(p^t-1)p^{n-t}+p^{m+\frac{d-t}{2}}$&$\frac{1}{2}p^{m-d}(p^t-1)(p^{n-d-t}-p^{\frac{n-d-t}{2}})(p^{n}-1)$
\\[2mm]
\hline

$(p^t-1)(p^{n-t}+p^{m+d-t})$&$p^{m-d-t}(p^{m-d}-1)(p^{m-d}-p^t+1)(p^{n}-1)\big{/}\left(p^{2d}-1\right)$\\[2mm]
\hline

$(p^t-1)p^{n-t}-p^{m+d-t}$&$p^{m-d-t}(p^t-1)(p^{n-2d}-1)(p^{n}-1)\big{/}\left(p^{2d}-1\right)$\\[2mm]
\hline
$(p^t-1)p^{n-t}$&$(p^{n}-1)(p^{3m-d}-p^{3m-2d}+p^{3m-3d}+p^{3m-2d-t}-p^{n-2d}+1)$
\\[2mm]
\hline

$0$&$1$
\\[2mm]
\hline
\end{tabular}
\end{center}

(ii). If $d'=d$ and $d/t$ is even, then
\begin{center}
\begin{tabular}{|c|c|}
\hline
$i$ & $A_i$ \\[2mm]
\hline
$(p^t-1)(p^{n-t}-p^{m-t})$&$p^{m+d-t}(p^{m}+p^t-1)(p^m-1)(p^m+1)^2\big{/}\left(2(p^d+1)\right)$
\\[2mm]
\hline
$(p^t-1)(p^{n-t}+p^{m-t})$&$p^{m+d-t}(p^{m}-p^t+1)(p^m-1)(p^{n}-2p^{n-d}+1)\big{/}\left({2(p^d-1)}\right)$
\\[2mm]
\hline

$(p^t-1)p^{n-t}+p^{m-t}$&$p^{m+d-t}(p^t-1)(p^n-1)^2\big{/}\left(2(p^d+1)\right)$
\\[2mm]
\hline

$(p^t-1)p^{n-t}-p^{m-t}$&$p^{m+d-t}(p^t-1)(p^n-1)(p^{n}-2p^{n-d}+1)\big{/}\left({2(p^d-1)}\right)$
\\[2mm]
\hline

 $(p^t-1)(p^{n-t}-p^{m+\frac{d}{2}-t})$&$\frac{1}{2}p^{m-d}(p^{n-d-t}+(p^t-1)p^{m-t-\frac{d}{2}})(p^{n}-1)$
\\[2mm]
\hline

$(p^t-1)(p^{n-t}+p^{m+\frac{d}{2}-t})$&$\frac{1}{2}p^{m-d}(p^{n-d-t}-(p^t-1)p^{m-t-\frac{d}{2}})(p^{n}-1)$
\\[2mm]
\hline

$(p^t-1)p^{n-t}-p^{m+\frac{d}{2}-t}$&$\frac{1}{2}p^{m-d}(p^t-1)(p^{n-d-t}+p^{m-t-\frac{d}{2}})(p^{n}-1)$
\\[2mm]
\hline

$(p^t-1)p^{n-t}+p^{m+\frac{d}{2}-t}$&$\frac{1}{2}p^{m-d}(p^t-1)(p^{n-d-t}-p^{m-t-\frac{d}{2}})(p^{n}-1)$
\\[2mm]
\hline

$(p^t-1)(p^{n-t}+p^{m+d-t})$&$p^{m-d-t}(p^{m-d}-1)(p^{m-d}-p^t+1)(p^{n}-1)\big{/}\left(p^{2d}-1\right)$\\[2mm]
\hline

$(p^t-1)p^{n-t}-p^{m+d-t}$&$p^{m-d-t}(p^t-1)(p^{n-2d}-1)(p^{n}-1)\big{/}\left(p^{2d}-1\right)$\\[2mm]
\hline
$(p^t-1)p^{n-t}$&$(p^{n}-1)(p^{3m-d}-p^{3m-2d}+p^{3m-3d}-p^{n-2d}+1)$
\\[2mm]
\hline

$0$&$1$
\\[2mm]
\hline
\end{tabular}
\end{center}

(iii). If $d'=2d$, then
\begin{center}
\begin{tabular}{|c|c|}
\hline
values & multiplicity \\[2mm]
\hline $(p^t-1)(p^{n-t}+p^{m-t})$&$
\begin{array}{ll}
&p^{m+3d-t}(p^m-p^t+1)(p^m-1)(p^{n}-p^{n-2d}-p^{n-3d}\\[1mm]
&\qquad+p^m-p^{m-d}+1)\big{/}\left((p^d+1)(p^{2d}-1)\right)
\end{array}
$
\\[2mm]
\hline $(p^t-1)p^{n-t}-p^{m-t}$&$
\begin{array}{ll}
&p^{m+3d-t}(p^t-1)(p^n-1)(p^{n}-p^{n-2d}-p^{n-3d}\\[1mm]
&\qquad+p^m-p^{m-d}+1)\big{/}\left((p^d+1)(p^{2d}-1)\right)
\end{array}
$
\\[2mm]

\hline
$(p^t-1)(p^{n-t}-p^{m+d-t})$&$p^{m-t}(p^{m-d}+p^t-1)(p^n-1)(p^m+p^{m-d}+p^{m-2d}+1)\big{/}(p^d+1)^2$
\\[2mm]
\hline
$(p^t-1)p^{n-t}+p^{m+d-t}$&$p^{m-t}(p^t-1)(p^{m-d}-1)(p^n-1)(p^m+p^{m-d}+p^{m-2d}+1)\big{/}(p^d+1)^2$
\\[2mm]
\hline

 $(p^t-1)(p^{n-t}+p^{m+2d-t})$&$p^{m-2d-t}(p^{m-2d}-p^t+1)
(p^{m-d}-1)(p^{n}-1)\big{/}\left((p^d+1)(p^{2d}-1)\right)$
\\[2mm]
\hline
 $(p^t-1)p^{n-t}-p^{m+2d-t}$&$p^{m-2d-t}(p^t-1)(p^{m-2d}-1)
(p^{m-d}-1)(p^{n}-1)\big{/}\left((p^d+1)(p^{2d}-1)\right)$
\\[2mm]

\hline
 $(p^t-1)p^{n-t}$&$
\begin{array}{ll}
 &(p^{n}-1)(p^{3m-d}-p^{3m-2d}+p^{3m-3d}-p^{3m-4d}+p^{3m-5d}\\[1mm]
&\qquad+p^{n-d}-2p^{n-2d}+p^{n-3d}-p^{n-4d}+1)
\end{array}
$
\\[2mm]
\hline $0$&$1$
\\[2mm]
\hline
\end{tabular}
\end{center}

\end{theo}
\begin{proof}
see \textbf{Appendix C.}
\end{proof}
\begin{remark}
The case (i) with $t=1$ has been shown in \cite{Zen Li}, Theorem 1.
\end{remark}

\section{Appendix A}

We need to introduce some results to prove Lemma \ref{rank}.

\begin{lemma}\label{num solution}(see Bluher \cite{Bluh}, Theorem 5.4 and 5.6)
Let $g(z)=z^{p^h+1}-bz+b$ with $b\in \bF_{p^l}^*$. Then the number
of the solutions to $g(z)=0$ in $\bF_{p^l}$ is $0$, $1$, $2$ or
$p^{\gcd(h,l)}+1$. Moreover, the number of $b\in \bF_{p^l}^*$ such
that $g(z)=0$ has unique solution in $\bF_{p^l}$ is
$p^{l-\gcd(h,l)}$ and if $z_0$ is the unique solution, then
$(z_0-1)^{\frac{p^l-1}{p^{\gcd(h,l)}-1}}=1.$
\end{lemma}

The following lemma has been proven in \cite{Bluh} and \cite{Zen
Li}. We will give some of the details for self-containing.
\begin{lemma}\label{psi}
Let $\psi_{\al,\be}(z)=\be^{p^{n-k}} z^{p^{m-k}+1}+\al z+\be$ with
$\al\in \bF_{p^m}^*, \be\in \bF_q^*$. Then
\begin{itemize}
  \item[(i).] $\psi_{\al,\be}(z)=0$ has either $0,1,2$ or $p^{d'}+1$
  solutions in $\bF_q$.
  \item[(ii).]  If $z_1, z_2$ are two solutions of
  $\psi_{\al,\be}(z)=0$ in $\bF_q$, then $z_1z_2$ is $(p^d-1)$-th power in $\bF_q$.
  \item[(iii).] If $\psi_{\al,\be}(z)=0$ has $p^{d'}+1$
  solutions in $\bF_q$, then for any two solutions $z_1$ and $z_2$, we
  have $z_1/z_2$ is a $(p^{d'}-1)$-th power in $\bF_q$.
  \item[(iv).] If $\psi_{\al,\be}(z)=0$ has exactly one
  solution in $\bF_q$, then it is a $(p^{d}-1)$-th power in $\bF_q$.
\end{itemize}
\end{lemma}
\begin{proof}
\begin{itemize}
  \item[(i).] By scaling $y=-\frac{\al}{\be} z$ and
$b=\frac{\al^{p^{m-k}+1}}{\beta^{p^{m-k}(p^m+1)}}$, we can rewrite
the equation $\psi_{\al,\be}(z)=0$ as
\begin{equation}\label{standard equ}
y^{p^{m-k}+1}-by+b=0.
\end{equation}
Since $b\in \bF_q^*$ and $\gcd(m-k, n)=\gcd(m-k,2k)=d'$, then the
result follows from Lemma \ref{num solution}.
  \item[(ii).] See \cite{Zen Li}, Prop.1 (2).
  \item[(iii).] Denote by $y_i=-\frac{\al}{\be} z_i$ for $i=1,2$. Since $\gcd(n,m-k)=d'$, from \cite{Bluh}, Theorem 4.6 (iv) we get $(y_1/y_2)^{\frac{q-1}{p^{d'}-1}}=1$ which is equivalent to
  $(z_1/z_2)^{\frac{q-1}{p^{d'}-1}}=1$.
  \item[(iv).] See \cite{Zen Li}, Prop.1 (3).
\end{itemize}
\end{proof}

{\it \textbf{Proof of Lemma \ref{rank}}}: (i). For
$Y=(y_1,\cdots,y_s)\in \bF_{q_0}^s$, $y=y_1v_1+\cdots+y_sv_s\in
\bF_q$, we know that
\begin{equation}\label{bil form1}
F_{\alpha,\beta}(X+Y)-F_{\alpha,\beta}(X)-F_{\alpha,\beta}(Y)=2XH_{\alpha,\beta}Y^T
\end{equation}
is equal to
\begin{equation}\label{bil form2}
f_{\alpha,\beta}(x+y)-f_{\alpha,\beta}(x)-f_{\alpha,\beta}(y)=\Tra_{d}^n\left(y(\alpha
x^{p^{m}}+\beta x^{p^k}+\be^{p^{n-k}} x^{p^{n-k}})\right)
\end{equation}
since $\Tra_d^m(\al x^{p^m}y+\al x y^{p^m})=\Tra_d^n(\al x^{p^m}y).$

 Let
\begin{equation}\label{def phi}
\phi_{\al,\be}(x)=\alpha x^{p^{m}}+\beta x^{p^k}+\be^{p^{n-k}}
x^{p^{n-k}}.
\end{equation}
 Therefore,
\[{\setlength\arraycolsep{2pt}
\begin{array}{lcl}
r_{\al,\be}=r& \Leftrightarrow&\text{the number of common solutions of}\;XH_{\alpha,\beta}Y^T=0\;\text{for all}\;Y\in \bF_{q_0}^s\;\text{is}\; q_0^{s-r}, \\[2mm]
& \Leftrightarrow&\text{the number of common solutions of}\;\Tra_{d}^n\left(y\cdot\phi_{\al,\be}(x)\right)=0\;\text{for all}\;y\in \bF_q\;\text{is}\; q_0^{s-r}, \\[2mm]
&\Leftrightarrow&\phi_{\al,\be}(x)=0\;\text{has}\; q_0^{s-r}\;
\text{solutions in}\; \bF_q.
\end{array}
}
\]

Since $\phi_{\al,\be}(x)$ is a $p^d$-linearized polynomial,  then
the set of the zeroes to $\phi_{\al,\be}(x)=0$ in $\bF_{p^n}$, say
$V$, forms an $\bF_{p^d}$-vector space.

If $\al=0$ and $\be\neq 0$, $\phi_{\al,\be}(x)=0$ becomes $\be
x^{p^k}+\be^{p^{n-k}}x^{p^{n-k}}=0$ and then $\be
^{p^k}x^{p^{2k}}+\be x=0$. In this case (\ref{def phi}) has $1$ or
$p^{d'}$ solutions according to $-\be^{1-p^k}$ is $(p^{d'}-1)$-th
power in $\bF_q$ or not. Hence $r_{0,\be}=s$ or $s-d'/d$. If
$\al\neq 0$ and $\be=0$, then $\phi_{\al,\be}(x)=0$ has unique
solution $x=0$ and as a consequence $r_{\al,0}=s$.

In the following we assume $\al\be\neq 0$, we need to consider the
nonzero solutions of $\phi_{\al,\be}(x)=0$. By substituting
$z=x^{p^k(p^{m-k}-1)}$ we get
\begin{equation}\label{def psi}
\psi_{\al,\be}(z)=\be^{p^{n-k}}z^{p^{m-k}+1}+\al z+\be=0.
\end{equation}

From Lemma \ref{num solution}, $\psi_{\al,\be}(z)=0$ has either
$0,1,2$ or $p^{d'}+1$ solutions in $\bF_q$.  In the case $d'=d$, by
Lemma \ref{psi}, if $\psi_{\al,\be}(z)=0$ has at least two solutions
in $\bF_q$, then all or none of the solutions  are $(p^d-1)$-th
power. Then $\psi_{\al,\be}(x)=0$ has $0, p^d-1, 2(p^d-1)$ or
$(p^d+1)(p^d-1)$ nonzero solutions. Take the  solution $x=0$ in
consideration, since $2p^d-1$ is not a $p^d$-th power, then it is
impossible and we get the result.

In the case $d'=2d$, the argument is almost the same except
$\psi_{\al,\be}(z)=0$ has two solutions $z_1,z_2$. If none, one or
two of the solutions is $(p^{2d}-1)$-th power, then
$\phi_{\al,\be}(x)=0$ has $1$, $p^{2d}-1$ or $2(p^{2d}-1)$ nonzero
solutions. But $2p^{2d}-1$ is not a $p^d$-th power. Then the result
follows.

In the case $d'=d$, if $\psi_{\al,\be}(z)=0$ has unique solution
$z_0\in \bF_q$, then it is also the unique solution in $\bF_{p^m}$
and the converse is also valid, since $b\in\bF_{p^m}$ and the
solutions of $\psi_{\al,\be}(z)=0$ in
$\bF_q\big{\backslash}\bF_{p^m}$ take on pairs $(z_0,z_0^{p^m})$. By
\cite{Bluh}, Theorem 5.6, the number of $b\in \bF_{p^m}^*$ such that
$\psi_{\al,\be}(z)=0$ has unique solution in $\bF_{p^m}$ is
$p^{m-d}$. For fixed $b$ and $\al\in \bF_{p^m}^*$, the number of
$\be\in \bF_{q}^*$ satisfying
$b=\frac{\al^{p^{m-k}+1}}{\beta^{p^{m-k}(p^m+1)}}$ is $p^m+1$. Hence
$n_1=p^{m-d}(p^m-1)(p^m+1)=p^{m-d}(p^n-1)$. $\square$

 {\it\textbf{Proof of Lemma \ref{moment}}}: (i). We observe that
\[ {\setlength\arraycolsep{2pt}
\begin{array}{ll}
&\sum\limits_{\al\in \bF_{p^m},\be\in
\bF_q}T(\al,\be)=\sum\limits_{\al\in
\bF_{p^m},\be\in\bF_q}\sum\limits_{x\in
\bF_q}\zeta_p^{\Tra_1^m(\al x^{p^m+1})+\Tra_1^n(\be x^{p^k+1})}\\[3mm]
&\quad\quad=\sum\limits_{x\in\bF_q}\sum\limits_{\al\in
\bF_{p^m}}\zeta_p^{\Tra_1^m(\al x^{p^m+1})}\sum\limits_{\be\in
\bF_q}\zeta_p^{\Tra_1^n(\be
x^{p^k+1})}=q\cdot\sum\limits_{\stackrel{\al\in
\bF_{p^m}}{x=0}}\zeta_p^{\Tra_1^m(\al x^{p^m+1})}=p^{3m}.
\end{array}
}
\]
(ii). We can calculate
\[
{ \setlength\arraycolsep{2pt}
\begin{array}{lll}
\sum\limits_{\al\in
\bF_{p^m},\be\in\bF_q}T(\al,\be)^2&=&\sum\limits_{x,y\in
\bF_q}\sum\limits_{\al\in
\bF_{p^m}}\zeta_p^{\Tra_1^m\left(\al\left(x^{p^m+1}+y^{p^m+1}\right)\right)}\sum\limits_{\be\in
\bF_q}\zeta_p^{\Tra_1^n\left(\be\left(x^{p^k+1}+y^{p^k+1}\right)\right)}\\[2mm]
&=&M_2\cdot p^{3m}\end{array}}
\] where
$M_2$ is the number of solutions to the equation
\begin{eqnarray}\label{def 2nd}
\left\{
\begin{array}{ll}
 x^{p^m+1}+y^{p^m+1}=0&\\[2mm]
 x^{p^k+1}+y^{p^k+1}=0&
 \end{array}
 \right.
\end{eqnarray}

If $xy=0$ satisfying (\ref{def 2nd}), then $x=y=0$. Otherwise
$(x/y)^{p^m+1}=(x/y)^{p^k+1}=-1$ which yields that
$(x/y)^{p^{m-k}-1}=1$. Denote by $x=ty$.  Since $\gcd(m-k,n)=d'$,
then $t\in \bF_{p^{d'}}^*$.
\begin{itemize}
  \item If $d'=d$, then $t\in \bF_{p^d}^*$ and (\ref{def 2nd}) is
  equivalent to $x^2+y^2=0$. Hence $t^2=-1$. There are two or none
  of $t\in \bF_{p^d}^*$ satisfying $t^2=-1$ depending on $p^d\equiv
  1\pmod 4$ or $p^d\equiv
  3\pmod 4$. Therefore
  \[
  M_2=\left\{
  \begin{array}{ll}
  1+2(q-1)=2q-1, &\text{if}\; p^d\equiv
  1\pmod 4\\[2mm]
  1, &\text{if}\; p^d\equiv
  3\pmod 4.
  \end{array}
  \right.
  \]
  \item If $d'=2d$, then by (\ref{rel d d'}) we get (\ref{def 2nd}) is equivalent to
  $x^{p^d+1}+y^{p^d+1}=0$. Then we have $t^{p^d+1}=-1$ which has
  $p^d+1$ solutions in $\bF_{p^{d'}}^*$. Therefore
  \[M_2=(p^d+1)(p^n-1)+1=p^{n+d}+p^n-p^d.\]
\end{itemize}

(iii). See \cite{Zen Li}, Prop.4 iii).$\square$

\begin{remark}
For the case $d'=2d$, $\sum\limits_{\al,\be\in \bF_q}T(\al,\be)^3$
can also be determined, but we do not need this result.
\end{remark}

 {\it\textbf{Proof of Lemma \ref{Artin}}}:
We get that
\[
\begin{array}{rcl}
qN&=&\sum\limits_{\om\in \bF_q}\sum\limits_{x,y\in
\bF_q}\zeta_p^{\Tra_1^n\left(\om\left(\frac{1}{2}\al x^{p^m+1}+\be x^{p^k+1}-y^{p^d}+y\right)\right)}\\[2mm]
&=&q^2+\sum\limits_{\om\in \bF_q^*}\sum\limits_{x\in
\bF_q}\zeta_p^{\Tra_1^n\left(\om\left(\frac{1}{2}\al x^{p^m+1}+\be
x^{p^k+1}\right)\right)} \sum\limits_{y\in
\bF_q}\zeta_p^{\Tra_1^n\left(y^{p^d}\left(\om^{p^d}-\om\right)\right)}\\[2mm]
&=&q^2+q\sum\limits_{\om\in \bF_{q_0}^*}\sum\limits_{x\in
\bF_q}\zeta_p^{\Tra_1^n\left(\om\left(\frac{1}{2}\al x^{p^m+1}+\be
x^{p^k+1}\right)\right)}\\[2mm]
&=&q^2+q\sum\limits_{\om\in \bF_{q_0}^*}\sum\limits_{x\in
\bF_q}\zeta_p^{\Tra_1^m\left(\om\al
x^{p^m+1}\right)+\Tra_1^n\left(\om\be x^{p^k+1}\right)}\\[2mm]
&=&q^2+q\sum\limits_{\om\in \bF_{q_0}^*}\sum\limits_{x\in
\bF_q}T(\om\al,\om\be)
\end{array}
\]
where the 3-rd equality follows from that the inner sum is zero
unless $\om^{p^d}-\om=0$, i.e. $\om\in \bF_{q_0}$ and the 4-th
equality follows from $\frac{1}{2}\om \al x^{p^m+1}\in \bF_{p^m}$.

For any $\om\in \bF_{q_0}^*$,  by (\ref{def H_al be}) we have
$F_{\om\al,\om\be}(X)=\om\cdot F_{\al,\be}(X)$,
$H_{\om\al,\om\be}=\om\cdot H_{\al,\be}$ and
$r_{\om\al,\om\be}=r_{\al,\be}$. From Lemma \ref{qua} (i) we know
that
\begin{equation}\label{om rel}
T({\om\al,\om\be})=\sum\limits_{X\in
\bF_{q_0}^s}\zeta_p^{\Tra_1^{d}(XH_{\om\al,\om\be}X^T)}=\eta_0(\om)^{r_{\al,\be}}T(\al,\be).
\end{equation}

In the case $d'=2d$, by Lemma \ref{rank}ii) we get that
$r_{\al,\be}$ is even. Hence $T(\om\al,\om\be)=T(\al,\be)$ for any
$\om\in \bF_{p^t}^*$ and $N=q+(p^d-1)T(\al,\be)$. $\square$

\section{Appendix B}

{\it\textbf{Proof of Theorem \ref{value dis T}}}:

Define
\[N_{i}=\left\{(\al,\be)\in
\bF_{p^m}\times
\bF_q\backslash\{(0,0)\}\left|r_{\al,\be}=s-i\right.\right\}.
\] Then $n_i=\big{|}N_i\big{|}$.

 According to Lemma \ref{qua} (setting
$F(X)=XH_{\al,\be}X^T=\Tra_{d}^m(\al x^{p^m+1})+\Tra_{d}^n(\be
x^{p^k+1})$), we define that for $\varepsilon=\pm 1$ and $0\leq
i\leq s-1$,
\[N_{i,\varepsilon}=\left\{
\begin{array}{ll}
&\left\{(\al,\be)\in \bF_{p^m}\times
\bF_q\backslash\{(0,0)\}\left|T(\al,\be)=\veps
p^{\frac{m+id}{2}} \right.\right\} \qquad\text{if}\; m+id\;\text{is even}\;,\\[3mm]
&\left\{(\al,\be)\in \bF_{p^m}\times
\bF_q\backslash\{(0,0)\}\left|T(\al,\be)=\veps\sqrt{p^*}
p^{\frac{m+id-1}{2}} \right.\right\} \qquad\text{if}\;
m+id\;\text{is odd}
\end{array}
\right.\] where $p^*=(-1)^{\frac{p-1}{2}}p$  and
$n_{i,\veps}=|N_{i,\veps}|$. Then $N_i=N_{i,1}\bigcup N_{i,-1}$ and
$n_i=n_{i,1}+n_{i,-1}$.

 (i). For the case $d'=d$, by
 Lemma \ref{rank} we have
\begin{equation}\label{par val}
n_{1,1}+n_{1,-1}=p^{m-d} (p^{n}-1) .
\end{equation}
Choose an element $\om\in \bF_{q_0}^*$ such that $\eta_0(\om)=-1$.
For any $(\al,\be)\in N_{1,\veps}$, since $s-1$ is odd, by (\ref{om
rel}) we get $T(\om\al,\om\be)=-T(\al,\be)$. Then the map
$(\al,\be)\mapsto (\om \al,\om\be)$ give a 1-to-1 correspondence
from $N_{1,1}$ to $N_{1,-1}$ Combining (\ref{par val}) one has
\begin{equation}\label{val n11 n1-1}
n_{1,1}=n_{1,-1}=\frac{1}{2}p^{m-d}(p^{n}-1).
\end{equation}

 Moreover, from Lemma \ref{moment}  and (\ref{val n11 n1-1}) we have
\begin{equation}\label{par sum1}
 \left(n_{0,1}-n_{0,-1}\right)+p^{d}\left(n_{2,1}-n_{2,-1}\right)=p^m(p^m-1)
 \end{equation}
\begin{equation}\label{par sum2}
 \left(n_{0,1}+n_{0,-1}\right)+p^{2d}\left(n_{2,1}+n_{2,-1}\right)=p^{n}(p^m-1)
 \end{equation}
\begin{equation}\label{par sum3}
 \left(n_{0,1}-n_{0,-1}\right)+p^{3d}\left(n_{2,1}-n_{2,-1}\right)=p^d(p^m-1)(-p^{n-d}+p^{m}+1).
 \end{equation}
In addition, by Lemma \ref{rank} and (\ref{val n11 n1-1}) we have
\begin{equation}\label{par sum0}
 \left(n_{0,1}+n_{0,-1}\right)+\left(n_{2,1}+n_{2,-1}\right)=(p^m-1)(p^{n}-p^{n-d}+p^m-p^{m-d}+1).
 \end{equation}
Combining (\ref{par sum1})--(\ref{par sum0}), together with
(\ref{val n11 n1-1}) we get the result.

 (ii). For the case $d'=2d$, by Lemma \ref{reduce num} we have
 \begin{equation}\label{val all 0}
 n_{0,1}=n_{2,-1}=n_{4,1}=0.
 \end{equation}
Combining Lemma \ref{rank}, Lemma \ref{moment} and (\ref{val all 0})
we have
\begin{equation}\label{par sum02}
 n_{0,-1}+n_{2,1}+n_{4,-1}=p^{3m}-1
 \end{equation}

\begin{equation}\label{par sum12}
 -n_{0,-1}+p^d\cdot n_{2,1}-p^{2d}\cdot n_{4,-1}=p^m(p^{m}-1)
 \end{equation}

 \begin{equation}\label{par sum22}
 n_{0,-1}+p^{2d}\cdot n_{2,1}+p^{4d}\cdot n_{4,-1}=p^m(p^{n+d}+p^{n}-p^m-p^d).
 \end{equation}
Solving the system of equations consisting of (\ref{par
sum02})--(\ref{par sum22}) yields the result. $\square$

{\it\textbf{Proof of Theorem \ref{wei dis C1}}}:
 From (\ref{Wei}) we know that for each non-zero codeword
$c(\al,\be)=\left(c_0,\cdots,c_{l-1}\right)$ $(l=q-1,
c_i=\Tra_1^m(\al\pi^{(p^m+1)i})+\Tra_1^n(\be \pi^{(p^k+1)i}), 0\leq
i\leq l-1, \text{and}\; (\al,\be)\in \bF_{p^m}\times\bF_q)$, the
Hamming weight of $c(\al,\be)$ is
\begin{equation}\label{wei c}
w_H\left(c(\al,\be)\right)=p^{m-t}(p^t-1)-\frac{1}{p^t}\cdot
R(\al,\be)
\end{equation}
where
\[R(\al,\be)=\sum\limits_{a\in \bF_{p^t}^*}T(a\al,a\be)=T(\al,\be)\sum\limits_{a\in \bF_{p^t}^*}\eta_0(a)^{r_{\al,\be}}\]
 by Lemma \ref{qua} (i).

 Let $\eta'$ be the quadratic (multiplicative) character on $\bF_q$.
 Then we have
\begin{enumerate}
    \item[(1).] if $d/t$ or $r_{\al,\be}$ is even, then $\sum\limits_{a\in \bF_{p^t}^*}\eta_0(a)^{r_{\al,\be}}=\sum\limits_{a\in \bF_{p^t}^*}
    1=p^t-1$ and $R(\al,\be)=(p^t-1)T(\al,\be)$.
    \item[(2).] if $d/t$  and $r_{\al,\be}$ are both odd, then $\sum\limits_{a\in \bF_{p^t}^*}\eta_0(a)^{r_{\al,\be}}=\sum\limits_{a\in \bF_{p^t}^*}
    \eta'(a)=0$ and $R(\al,\be)=0$.
\end{enumerate}

Thus the weight distribution of $\cC_1$ can be derived from Theorem
\ref{value dis T} and (\ref{wei c}) directly. For example, if  $d/t$
is odd and $d'=d$, then
\begin{itemize}
    \item[(1).] if
    $r_{\al,\be}=s$ and $T(\al,\be)=p^{m}$, then
    $w_H(c(\al,\be))=(p^t-1)(p^{n-t}-p^{m-t})$.
    \item[(2).]  if
    $r_{\al,\be}=s$ and $T(\al,\be)=-p^{m}$, then
    $w_H(c(\al,\be))=(p^t-1)(p^{n-t}+p^{m-t})$.
    \item[(3).] if $r_{\al,\be}=s-1$, then
    $w_H(c(\al,\be))=(p^t-1)p^{n-t}$.
    \item[(4).] if $r_{\al,\be}=s-2$ and $T(\al,\be)=-p^{m+d}$, then
    $w_H(c(\al,\be))=(p^t-1)(p^{n-t}+p^{m+d-t})$.
\end{itemize} $\square$

\section{Appendix C}

{\it\textbf{Proof of Lemma \ref{last lemma}}}:

Define $n(\al,\be,a)$ to be the number of $\ga\in \bF_q$ satisfying
(i) and (ii). From (\ref{def H_al be}) we know that
$XH_{\al,\be}X^T=\Tra_{d}^m(\al x^{p^m+1})+\Tra_{d}^n(\be
x^{p^k+1})$. Combining (\ref{def A_gamma}), (\ref{bil form1}) and
(\ref{bil form2}) we can get
\begin{eqnarray}\label{linear equivalent}
{\setlength\arraycolsep{2pt}
\begin{array}{lcl}
2XH_{\al,\be}+A_{\ga}=0& \Leftrightarrow&\;2XH_{\al,\be}Y^{T}+A_{\ga}Y^T=0\;\text{for all}\;Y\in \bF_{q_0}^s\\[2mm]
& \Leftrightarrow&\;\Tra_{d}^n\left(y\phi_{\al,\be}(x)\right)+\Tra_{d}^n(\ga y)=0\;\text{for all}\;y\in \bF_q \\[2mm]
& \Leftrightarrow&\;\Tra_{d}^n\left(y(\phi_{\al,\be}(x)+\ga)\right)=0\;\text{for all}\;y\in \bF_q \\[2mm]
&\Leftrightarrow&\phi_{\al,\be}(x)+\ga=0.
\end{array}
}
\end{eqnarray}

Let $x_0$, $x_0'$ be two distinct solutions of (i) (if exists). We
can get $x_0=X_0 \cdot V^T$ and $x'_0=X'_0\cdot V^T$ with $X_0,
X'_0\in \bF_{q_0}^s$ and $V=(v_1,\cdots, v_n)$. Define $\Delta
X_0=X'_0-X_0$ and $\Delta x_0=x'_0-x_0=X_0\cdot V^T$. Then
\[\phi_{\al,\be}(x_0)+\ga=\phi_{\al,\be}(x'_0)+\ga=0\]
gives us
\[2X_0H_{\al,\be}+A_{\ga}=2X'_0H_{\al,\be}+A_{\ga}=0\]
and hence
\[\Delta X_0\cdot H_{\al,\be}=0.\]
It follows that
\[
\begin{array}{ll}
 &X'_0\cdot  H_{\al,\be}\cdot {X'_0}^T=(X_0+\Delta X_0)\cdot H_{\al,\be}\cdot (X_0+\Delta X_0)^T\\[2mm]
 &\qquad=X_0  H_{\al,\be} {X}_0^T+ \Delta X_0\cdot H_{\al,\be}\cdot (\Delta X_0+2 X_0)=X_0  H_{\al,\be} {X}_0^T.
 \end{array}
 \]

 Therefore
\[
\begin{array}{ll}
&\Tra_{t}^m(\al {x'_0}^{p^m+1})+\Tra_{t}^n(\be
{x'_0}^{p^k+1})=\Tra_{t}^{d}\left(\Tra_{t}^m(\al
{x'_0}^{p^m+1})+\Tra_{t}^n(\be
{x'_0}^{p^k+1})\right)=\Tra_{t}^{d}\left(X'_0\cdot
H_{\al,\be}\cdot {X'_0}^T\right)\\[2mm]
&=\Tra_{t}^{d}\left(X_0 H_{\al,\be}
{X_0}^T\right)=\Tra_{t}^{d}\left(\Tra_{t}^m(\al
{x_0}^{p^m+1})+\Tra_{t}^n(\be {x_0}^{p^k+1})\right)=\Tra_{t}^m(\al
{x_0}^{p^m+1})+\Tra_{t}^n(\be {x_0}^{p^k+1}).
\end{array}
\]
Hence $n(\al,\be,a)$ is well-defined (independent of the choice of
$x_0$).

If (i) is satisfied,  that is, $\phi_{\al,\be}(x)+\ga=0$ has
solution(s) in $\bF_q$ which yields that $2XH_{\al,\be}+A_{\ga}=0$
has solution(s). Note that $\mathrm{rank}\, H_{\al,\be}=s-i$.
Therefore $2XH_{\al,\be}+A_{\ga}=0$ has $q_0^i=p^{id}$ solutions
with $X\in \bF_{q_0}^s$ which is equivalent to saying
$\phi_{\al,\be}(x)+\ga=0$ has $p^{id}$ solutions in $\bF_q$.
Conversely, for any $x_0\in \bF_q$, we can determine $\ga$ by
$\ga=-\phi_{\al,\be}(x_0).$ Let $N(\al,\be,a)$ be the number of
$x_0\in \bF_q$ satisfying (ii). Then we have
$n(\al,\be,a)=N(\al,\be,a)\big{/}p^{id}$.

Let $\chi'(x)=\zeta_p^{\Tra_1^{t}(x)}$ with $x\in \bF_{p^t}$ be an
additive character on $\bF_{p^t}$ and
$G(\eta',\chi')=\sum\limits_{x\in \bF_{p^t}}\eta'(x)\chi'(x)$ be the
Gaussian sum on $\bF_{p^t}$. We can calculate
\[
\begin{array}{rcl}
p^t\cdot N(\al,\be,a)&=&\sum\limits_{x\in \bF_q}\sum\limits_{\om\in
\bF_{p^t}}\zeta_p^{\Tra_1^{t}\left(\om\cdot\left( \Tra_{t}^m(\al
x^{p^m+1})+\Tra_{t}^n(\be
x^{p^k+1})-a\right)\right)}\\[3mm]
&=&p^n+\sum\limits_{\om \in \bF_{p^t}^*}T(\om \al,\om \be)\zeta_p^{-\Tra_1^{t}(a\om)}\\[2mm]
&=&p^n+T(\al,\be)\cdot\sum\limits_{\om \in
\bF_{p^t}^*}\eta_0(\om)^{s-i}\chi'(-a\om)
\end{array}
\]
where the 3-rd equality holds from (\ref{om rel}) for any $\om \in
\bF_{p^t}^*\subset \bF_{q_0}^*$.
\begin{itemize}
    \item If $s-i$ and $d/t$ are both odd, and $a=0$, then
    $\eta_0(\om)^{s-i}=\eta'(\om)$ and $N(\al,\be,0)=p^{n-t}$.
    \item If $s-i$ and $d/t$ are both odd, and $a\neq 0$, then
     \[
\begin{array}{rcl}
N(\al,\be,a)&=&p^{n-t}+\frac{1}{p^t}\cdot
T(\al,\be)\cdot\sum\limits_{\om \in
\bF_{p^t}^*}\eta_0(\om)\chi'(-a\om)\\[2mm]
&=&p^{n-t}+\frac{1}{p^t}\cdot T(\al,\be)\cdot\eta'(-a)\cdot G(\eta',\chi')\\[2mm]
&=&p^{n-t}+\veps \eta'(a)p^{\frac{n+id-t}{2}}
\end{array}
\]
where the 2-nd equality follows from the explicit evaluation of
quadratic Gaussian sums (see \cite{Lid Nie}, Theorem 5.15 and 5.33).
    \item If $s-i$ or $d/t$ is even, and $a=0$, then $\eta_0(\om)^{s-i}=1$
    for any $\om\in \bF_{p^t}^*$ and $N(\al,\be,0)=p^{n-t}+
    \veps (p^t-1)p^{\frac{n+id}{2}-t}$.
    \item If $s-i$ or $d/t$ is even, and $a\neq 0$, then
    \[
\begin{array}{rcl}
N(\al,\be,a)&=&p^{n-t}+\frac{1}{p^t}\cdot T(\al,\be)\cdot \sum\limits_{\om\in \bF_{p^t}^*}\chi'(-a\om)\\[2mm]
&=&p^{n-t}-\veps p^{\frac{n+id}{2}-t}.
\end{array}
\]
\end{itemize}
Therefore we complete the proof by dividing $p^{id}$. $\square$

{\it\textbf{Proof of Theorem \ref{value dis S}}}: Define
\[\Xi=\left\{(\al,\be,\ga)\in \bF_q^3\left|S(\al,\be,\ga)=0 \right.\right\}\]
and $\xi=\big{|}\Xi\big{|}$.

 Recall $n_i,H_{\al,\be},,r_{\al,\be},A_{\ga}$ in Section 1 and
$N_{i,\veps},n_{i,\veps,}$ in the proof of Lemma \ref{rank}. Note
that $2XH_{0,0}+A_{\ga}=0$ is solvable if and only if $\ga=0$. If
$(\al,\be)\in N_{i,\veps}$, then the number of $\ga\in \bF_q$ such
that $2XH_{\al,\be}+A_{\ga}=0$ is solvable is $q_0^{s-i}=p^{n-id}$.
 From
Lemma \ref{rank} (i) we know that
\begin{itemize}
  \item if $d'=d$ and $(\al,\be)\neq (0,0)$, then
$r_{\al,\be}=s-i$ for some $i\in\{0,1,2\}$. By Lemma \ref{qua} (ii)
we have
\begin{equation}\label{value xi1}
{\setlength\arraycolsep{2pt}
\begin{array}{lll}
\xi&=&p^{n}-1+(p^{n}-p^{n-d})n_{1}+(p^{n}-p^{n-2d})n_2\\[2mm]
&=&(p^{n}-1)(p^{3m-d}-p^{3m-2d}+p^{3m-3d}-p^{n-2d}+1).
\end{array}
}
\end{equation}
  \item if $d'=2d$, similarly we have
\begin{equation}\label{value xi2}
{\setlength\arraycolsep{2pt}
\begin{array}{lll}
\xi&=&p^{n}-1+(p^{n}-p^{n-2d})n_{2,1}+(p^{n}-p^{n-4d})n_{4,-1}\\[2mm]
&=&(p^{n}-1)(p^{3m-d}-p^{3m-2d}+p^{3m-3d}-p^{3m-4d}+p^{3m-5d}\\[1mm]
&&\qquad+p^{n-d}-2p^{n-2d}+p^{n-3d}-p^{n-4d}+1).
\end{array}
}
\end{equation}
\end{itemize}

Assume $(\al,\be)\in N_{i,\veps}$ and $\phi_{\al,\be}(x)+\ga=0$ has
solution(s) in $\bF_q$ (choose one, say $x_0$). Then by Lemma
\ref{qua} we get
\[S(\al,\be,\ga)=\zeta_p^{-\left(\Tra_t^m (\alpha
x_0^{p^{m}+1})+\Tra_t^n(\beta x_0^{p^k+1})\right)}\cdot
T(\al,\be).\]

Applying Lemma \ref{last lemma} for $t=1$ and Theorem \ref{value dis
T}, we get the result. $\square$

{\it\textbf{Proof of Theorem \ref{cor dis}}}: Recall
$M_{(\al_1,\be_1),(\al_2,\be_2)}(\tau)$ defined in (\ref{cor fun})
and (\ref{coe cor}).
 Fix $(\al_2,\be_2)\in
\bF_{p^m}\times \bF_q$, when $(\al_1,\be_1)$ runs through
$\bF_{p^m}\times \bF_q$ and $\tau$ takes value from $0$ to $q-2$,
$(\al',\be',\ga')$ runs through $\bF_{p^m}\times
\bF_q\times\left\{\bF_{q}\big{\backslash}\{1\}\right\}$ exactly one
time.

For any possible value $\kappa$ of $S(\al,\be,\ga)$, define

\[s_{\kappa}=\#\left\{(\al,\be,\ga)\in \bF_{p^m}\times \bF_q\times \bF_q\,\displaystyle{|}\,S(\al,\be,\ga)=\kappa\right\}\]

\[s^1_{\kappa}=\#\left\{(\al,\be,\ga)\in \bF_{p^m}\times \bF_q\times \left\{\bF_q\backslash\{1\}\right\}\,\big{|}\,S(\al,\be,\ga)=\kappa\right\}\]

and
\[t_{\kappa}=\#\left\{(\al,\be)\in \bF_{p^m}\times \bF_q\,\displaystyle{|}\,T(\al,\be)=\kappa\right\}.\]

By Lemma \ref{q-2} we have
\[s_{\kappa}^1=\frac{q-2}{q-1}\times (s_{\kappa}-t_{\kappa})+t_{\kappa}=\frac{q-2}{q-1}\times s_{\kappa}+\frac{1}{q-1}\times t_{\kappa}.\]

Define $M_{\kappa}$ to be the number of $(\al_1,\be_1,\al_2,\be_2)$
such that $M_{(\al_1,\be_1),(\al_2,\be_2)}=\kappa$. Hence we get
\[M_{\kappa}=p^{3m}\cdot s_{\kappa}^1=p^{3m}\cdot\left(\frac{q-2}{q-1}\cdot s_{\kappa}+\frac{1}{q-1}\cdot t_{\kappa}\right).\]

Then the result follows from Theorem \ref{value dis T} and Theorem
\ref{value dis S}. $\square$

{\it\textbf{Proof of Theorem \ref{wei dis C2}}}: From (\ref{Wei}) we
know that for each non-zero codeword
$c(\al,\be,\ga)=\left(c_0,\cdots,c_{q-2}\right)$ $(
c_i=\Tra_{t}^m(\al\pi^{(p^m+1)i})+\Tra_{t}^n(\be \pi^{(p^k+1)i}+\ga
\pi^{i}),\, 0\leq i\leq q-2,\, \text{and}\; (\al,\be,\ga)\in
\bF_{p^m}\times \bF_q^2)$, the Hamming weight of $c(\al,\be,\ga)$ is
\begin{equation}\label{wei c2}
w_H\left(c(\al,\be,\ga)\right)=p^{n-t}(p^t-1)-\frac{1}{p^t}\cdot
R(\al,\be,\ga)
\end{equation}
where
\[R(\al,\be,\ga)=\sum\limits_{\om\in \bF_{p^t}^*}S(\om\al,\om\be,\om\ga).\]

For any $\om\in \bF_{p^t}^*\subset \bF_{q_0}^*$, we have
$\phi_{\om\al,a\be}(x)+\om\ga=0$ is equivalent to
$\phi_{\al,\be}(x)+\ga=0$. Let $x_0\in \bF_q$ be a solution of
$\phi_{\al,\be}(x)+\ga=0$ (if exist).
\begin{itemize}
    \item[(1).]
If $\phi_{\al,\be}(x)+\ga=0$ has solutions in $\bF_q$,  then by
Lemma \ref{qua} and (\ref{om rel}) we have
\[
\begin{array}{ll}
&S(\om\al,\om\be,\om\ga)=\zeta_p^{-\left(\Tra_1^m(\om\al
x_0^{p^m+1})+\Tra_1^n(\om\be
x_0^{p^k+1})\right)}T(\om\al,\om\be)\\[1mm]
&\qquad=\zeta_p^{-\left(\Tra_1^m(\om\al x_0^{p^m+1})+\Tra_1^n(\om\be
x_0^{p^k+1})\right)}\eta_0(\om)^{r_{\al,\be}}T(\al,\be).
\end{array}
\]
 Hence
\begin{equation}\label{value of R}
R(\al,\be,\ga)=T(\al,\be)\sum\limits_{\om\in
\bF_{p^t}^*}\zeta_p^{-\Tra_1^t\left(\om\cdot\left(\Tra_t^m(\al
x_0^{p^m+1})+\Tra_t^n(\be
x_0^{p^k+1})\right)\right)}\eta_0(\om)^{r_{\al,\be}}.\nonumber
\end{equation}
    Fix $(\al,\be)\in N_{i,\veps}$ for $\veps=\pm 1$,
and suppose $\phi_{\al,\be}(x)+\ga=0$ is solvable in $\bF_q$. Denote
by $\vartheta=\Tra_t^m(\al x_0^{p^m+1})+\Tra_t^n(\be x_0^{p^k+1})$.
Then
\begin{itemize}
    \item if $s-i$ and $d/t$ are both odd, and $\vartheta=0$, then
\[R(\al,\be,\ga)=T(\al,\be)\sum\limits_{\om\in
\bF_{p^t}^*}\eta'(\om)=0.\]
    \item if $s-i$ and $d/t$ are both odd, and $\vartheta\neq 0$, then by the result of
quadratic Gaussian sums
\[
\begin{array}{rcl}
R(\al,\be,\ga)&=&T(\al,\be)\eta'(-\vartheta)G(\eta',\chi')\\[2mm]
&=&\veps\eta'(\vartheta)p^{\frac{n+id+t}{2}},\\[2mm]
&=&\left\{
\begin{array}{ll}
p^{m+\frac{id+t}{2}}&\text{if}\; \veps=\eta'(\vartheta),\\[2mm]
-p^{m+\frac{id+t}{2}}&\text{if}\; \veps=-\eta'(\vartheta).
\end{array}
\right.
\end{array}
\]

    \item if $s-i$ or $d/t$ is even, and $\vartheta=0$, then $\eta_0(\om)^{r_{\al,\be}}=1$ for $\om\in \bF_{p^t}^*$ and
$R(\al,\be,\ga)=(p^t-1)T(\al,\be)=\veps(p^t-1)p^{m+\frac{id}{2}}$.
    \item if $s-i$ or $d/t$ is even, and $\vartheta\neq 0$, then
    $\eta_0(\om)^{r_{\al,\be}}=1$ for $\om\in \bF_{p^t}^*$
and $R(\al,\be,\ga)=-T(\al,\be)=-\veps p^{m+\frac{id}{2}}$.
\end{itemize}
    \item[(2).] If $\phi_{\al,\be}(x)+\ga=0$ has no solutions in
    $\bF_q$ which implies that
$\phi_{\om\al,\om\be}(x)+\om\ga=0$ also has no solutions in $\bF_q$
for any $\om\in \bF_{p^t}^*\subset \bF_{q_0}$. Hence
$S(\om\al,\om\be,\om\ga)=0$ and $R(\al,\be,\ga)=0$.
\end{itemize}

 Thus the weight distribution of $\cC_2$ can be derived from Theorem \ref{value dis T},  Lemma
 \ref{last lemma}, (\ref{value xi1}), (\ref{value xi2})
 and (\ref{wei c2}) directly.
$\square$

\section{Conclusion}

\quad In this paper we have studied the exponential sums
$\sum\limits_{x\in \bF_q}\zeta_p^{\Tra_1^m (\alpha
x^{p^{m}+1})+\Tra_1^n(\beta x^{p^k+1})}$ and $\sum\limits_{x\in
\bF_q}\zeta_p^{\Tra_1^m (\alpha x^{p^{m}+1})+\Tra_1^n(\beta
x^{p^k+1}+\ga x)}$ with $\al\in \bF_{p^m},(\be,\ga)\in \bF_q^2$.
After giving the value distribution of $\sum\limits_{x\in
\bF_q}\zeta_p^{\Tra_1^m (\alpha x^{p^{m}+1})+\Tra_1^n(\beta
x^{p^k+1})}$ and $\sum\limits_{x\in \bF_q}\zeta_p^{\Tra_1^m (\alpha
x^{p^{m}+1})+\Tra_1^n(\beta x^{p^k+1}+\ga x)}$, we determine the
correlation distribution among a family of sequences, and the weight
distributions of the cyclic codes $\cC_1$ and $\cC_2$. These results
generalize  \cite{Zen Li}.

\section{Acknowledgements}
\quad The authors will thank the anonymous referees for their
helpful comments.

\end{document}